\documentclass[superscriptaddress,aps,pra,twocolumn, showpacs,nofootinbib,longbibliography,notitlepage]{revtex4-2}

\usepackage[utf8]{inputenc}
\usepackage[T1]{fontenc}     %Output what you want e.g., é, ł, a, ü
\usepackage[british]{babel}  %Do hyphenation according to british english
\usepackage[sc,osf]{mathpazo}
\usepackage[scaled=0.86]{berasans}  % URL font that go well wtih palatino
\usepackage[colorlinks=true, allcolors=blue, urlcolor=blue]{hyperref}  %Hyperlinks (pink, green, blue)
\usepackage{graphicx} % Package to insert exteral figures
\usepackage[babel]{microtype}  %Improves text justification
\usepackage{amsmath,amssymb,amsthm,bm,amsfonts,mathrsfs,bbm} %Usefull math packages

\usepackage{xspace}  %Useful to add space in macros
\usepackage{pgfplots}
\usepackage{xcolor,colortbl}
\usepackage{array}
\usepackage{bigstrut}
\usepackage{ulem}
\usepackage{mathtools}
\usepackage{physics}
\usepackage{caption}
\usepackage{subcaption}
\usepackage{tabularx}
\newcolumntype{C}{>{\centering\arraybackslash}X}
\usepackage{lipsum}
\usepackage[inkscapelatex=true]{svg}
\usepackage{hyperref}
\hypersetup{
	colorlinks=true,
	linkcolor=red,
	filecolor=blue,      
	urlcolor=blue,
	citecolor=purple,
}

\newcounter{subeq}

\newcommand{\be}{\begin{equation}}
	\newcommand{\ee}{\end{equation}}
\newcommand{\bea}{\begin{eqnarray}}
	\newcommand{\eea}{\end{eqnarray}}
\newcommand{\bes}{\begin{equation*}}
	\newcommand{\ees}{\end{equation*}}
\newcommand{\beas}{\begin{eqnarray*}}
	\newcommand{\eeas}{\end{eqnarray*}}

\newtheorem{thm}{Theorem}
\newtheorem*{thm*}{Theorem}

\newtheorem{lem}[thm]{Lemma}
\newtheorem*{lem*}{Lemma}
\newtheorem{prop}{Proposition}
\newtheorem*{prop*}{Proposition}

\newtheorem*{lipschitzLem*}{Lemma \ref{lipschitz}}
\newtheorem*{lipschitzCubeLem*}{Lemma \ref{lipschitzCube}}
\newtheorem*{pgmNearlyOptimalThm*}{Theorem \ref{pgmNearlyOptimal}}

\begin{document}

	\title{Self-testing of Nonmaximal Genuine Entangled States using Tripartite Hardy Relations}

	\author{Ranendu Adhikary}
	\email{ronjumath@gmail.com}
	\affiliation{Cryptology and Security Research Unit, Indian Statistical Institute, 203 B.T. Road, Kolkata 700108, India.}
	
	\author{Souradeep Sasmal}
	\email{souradeep.007@gmail.com}
	\affiliation{Institute of Fundamental and Frontier Sciences,
		University of Electronic Science and Technology of China, Chengdu 611731, China}
	
	\author{Arup Roy}
	\email{arup145.roy@gmail.com}
	\affiliation{Department of Physics, Hooghly Mohsin College, West Bengal 712 101, India}

	%%%%%%%%%%%%%%%%%%%%%%%%%%%%%%%%%%%%%%%%%%%%%%%%%%%%%%%%%%%%%%%%%%%%%%%%%%%%%%%%%%%%%%%%%%%%%%%%%%%%%%%%%%%%%%%%%%%%%%%%%%%%%%%%%%%%%%%%%%%%%%%%%%%%%%%%%%%%%%%%%%%%%%%%%%%%%%%%%%%%%%%%%%%%%%%%%%%%%%%%%%%%%%%%%%%%%%%%%%%%%%%%%%%%%%%%%%
	
	\begin{abstract}
		We demonstrate that, in the tripartite scenario with all parties' local events being space-like separated, Hardy-type nonlocality constitutes a stronger manifestation of nonlocal correlations than those captured by Mermin-type inequalities, an important distinction that has hitherto remained unrecognised. To substantiate this assertion, we develop a general framework for the characterisation of tripartite correlations by extending the notion of Settings Independence and Outcome Independence beyond their bipartite formulation. This framework highlights the pivotal role of Hardy-type reasoning in the detection and certification of genuine multipartite nonlocality. Furthermore, we show that the tripartite Hardy-nonlocality enables the self-testing of a broad class of pure nonmaximally genuine entangled tripartite states. A key advantage of Hardy-based self-testing over methods based on tripartite Bell inequalities is its ability to certify quantum correlations even in the presence of nonmaximal violations. This, in turn, facilitates the device-independent certification of randomness from Hardy-type correlations. Unlike Bell functionals, which typically enable self-testing of only a single extremal point per inequality, Hardy relation self-tests a set of extremal quantum correlations for any nonzero Hardy probability. We find that the maximum certifiable randomness using Hardy-type correlations is $\log_2 7\approx 2.8073$-bits, highlighting both the practical  and foundational significance of Hardy-based techniques for quantum randomness generation.
	\end{abstract}
	
	\maketitle
	
	%%%%%%%%%%%%%%%%%%%%%%%%%%%%%%%%%%%%%%%%%%%%%%%%%%%%%%%%%%%%%%%%%%%%%%%%%%%%%%%%%%%%%%%%%%%%%%%%%%%%%%%%%%%%%%%%%%%%%%%%%%%%%%%%%%%%%%%%%%%%%%%%%%%%%%%%%%%%%%%%%%%%%%%%%%%%%%%%%%%%%%%%%%%%%%%%%%%%%%%%%%%%%%%%%%%%%%%%%%%%%%%%%%%%%%%%%%

	\section{Introduction} \label{intro}
	
	In practical quantum experiments, a key challenge lies in verifying whether the observed statistics genuinely originate from the intended quantum state and measurements, particularly when the devices are untrusted. Self-testing, based on the Bell nonlocality, addresses this by enabling the inference of both quantum states and measurements directly from observed correlations (i.e., joint probability distributions), without any assumptions regarding the internal workings of the devices \cite{Supic2020}. As such, self-testing underpins many applications in device-independent (DI) quantum information processing, including cryptography, randomness generation, delegated quantum computation, and foundational studies of quantum theory.
	
	Bell nonlocality \cite{Bell1964} is typically witnessed through the violation of suitable Bell inequalities, which are linear functionals defined over the set of joint probability distributions \cite{Brunner2014}. Within the quantum set, the optimal violations corresponds to unique extremal correlations, realised by specific quantum strategies (up to local unitaries and complex conjugation). Hence, optimal Bell violations form the basis for self-testing \cite{Supic2020}. However, many extremal quantum correlations are non-exposed; they do not maximise any Bell functional and thus cannot be self-tested through Bell inequality violations. These non-exposed extremal points often originate from nonmaximally entangled states, exposing a fundamental limitation of linear Bell inequalities.
	
	To overcome this, alternative approaches characterising the full quantum behaviour have been proposed \cite{Barizien2024}. Of particular interest are Hardy-type nonlocal arguments \cite{Hardy1992, Cabello2002, Liang2005}. Hardy's argument demonstrates that if certain events are constrained to have zero probability within a local ontic (realist) model, then another event, the Hardy probability, must also necessarily vanish \cite{Hardy1992}. Yet, quantum theory predicts a strictly nonzero value for specific states \cite{Hardy1992}, resulting in a direct contradiction with the predictions of any local realist model. It has been shown that Hardy nonlocality enables self-testing \cite{Rabelo2012, Rai2021, Rai2022} of non-exposed extremal correlations on the nosignalling boundary \cite{Goh2018, Rai2019, Chen2023}, even without maximal violation, thereby certifying a broader class of quantum correlations than those accessible via Bell inequalities.
	
	Despite this promise, Hardy-based self-testing has so far been largely confined to bipartite systems. In particular, the self-testing of tripartite systems remains significantly less explored. Early developments in multipartite self-testing concentrated on graph states \cite{Baccari2020}, the maximally entangled GHZ state \cite{Sarkar2022, Panwar2023, Singh2024}, and subsequently on partially entangled GHZ states via stabiliser-based approaches \cite{Supic2018}. The $W$ state \cite{Dur2000}, another prominent multipartite entangled state, has been self-tested using the SWAP method \cite{Wu2014}. Nevertheless, progress in this area has been constrained by the intrinsic complexity of multipartite systems \cite{Gallego2012, Bancal2013} and the rich variety of pure tripartite entangled states, each characterised by distinct entanglement structures. Further challenges arise from issues such as bi-separability and the reliable identification of genuine multipartite entanglement \cite{Acin2000}, complicating the self-testing landscape even further.
	
	A major conceptual challenge in multipartite scenarios is that the definition of locality itself becomes significantly more nuanced \cite{Svetlichny1987, Gallego2012, Bancal2013, Dutta2020}. For instance, in a setting with three spatially separated parties (Alice, Bob, and Charlie), generating deterministic outputs from local inputs can exhibit fully-local correlations \cite{Mermin1990}. If no causal constraints are imposed between two of the parties (e.g., Bob and Charlie), the correlations reduce to bilocal forms \cite{Svetlichny1987, Bancal2013}. Any correlation that cannot be decomposed into such a bilocal form is said to exhibit genuine multipartite nonlocality.
	
	In this work, we propose a framework for characterising multipartite nonlocality based on the notions of Settings Independence (SI) and Outcome Independence (OI), generalising the seminal ideas of Jarrett and Shimony from the bipartite case \cite{Jarrett1984, Shimony1993}. This framework not only provides a clearer conceptual foundation for understanding nonlocality in tripartite systems, but also captures essential causal independence constraints across multiple parties. Specifically, we show that this approach brings out a stronger form of tripartite nonlocality exhibited by the tripartite Hardy relation proposed in \cite{Rahaman2014} than was previously identified.
	
	We begin by showing that fully local correlations satisfy mutual SI and OI between all parties, implying a vanishing Hardy probability. Remarkably, we find that even if two-way OI is relaxed between any pair (e.g., Bob and Charlie), the Hardy probability remains zero. This leads us to conclude that that a nonzero Hardy probability necessarily implies the failure of both two-way SI across all parties and two-way OI between at least one pair. Moreover, we show that such correlations correspond precisely to nosignalling genuine nonlocal correlations, as characterised in \cite{Bancal2013}, thereby enabling the certification of genuine multipartite nonlocality in scenarios where the local events of Alice, Bob, and Charlie are space-like separated.
	
	Building upon this, we leverage Hardy-type genuine nonlocality to self-test a broad class of pure, genuinely entangled tripartite quantum states, including scenarios where the Hardy relation is nonmaximally violated. Furthermore, we demonstrate that such self-testing facilitates the certification of genuine quantum randomness from observed Hardy statistics, highlighting both the foundational significance and the practical utility of our results.
	
	The paper is organised as follows. In Sec.~\ref{hsioi}, we introduce the notion of OI and SI and establish that tripartite Hardy nonlocal correlations are genuine nonlocal. Sec.~\ref{hardybehav} characterises all pure three-qubit states compatible with the Hardy-type nonlocality conditions. In Sec.~\ref{stth}, we show how to self-test quantum states that exhibit nonmaximal Hardy violations. Sec.~\ref{randomness} derives the maximum certifiable randomness arising from Hardy-type correlations. We conclude with a discussion of the implications and future directions in Sec.~\ref{conclu}.

	%%%%%%%%%%%%%%%%%%%%%%%%%%%%%%%%%%%%%%%%%%%%%%%%%%%%%%%%%%%%%%%%%%%%%%%%%%%%%%%%%%%%%%%%%%%%%%%%%%%%%%%%%%%%%%%%%%%%%%%%%%%%%%%%%%%%%%%%%%%%%%%%%%%%%%%%%%%%%%%%%%%%%%%%%%%%%%%%%%%%%%%%%%%%%%%%%%%%%%%%%%%%%%%%%%%%%%%%%%%%%%%%%%%%%%%%%%%%%%%%%%%%%%%%%%%%%%%%%%%%%%%%%%%%%%%%%%%%%%%%%%%%%%%%%%%%%%%%%%%%%%%%%%%%%%%%%%%%%%%%%%%%%%%%%%%%%%%%%%%%%%%%%%%%%%%%%%%%%%%%%%%%%%%%%%%%%%%%%%%%%%%%%%%%%%%%%%%%%%%%%%%%%%%%%%%%%%%%%%%%%%%%%%%%%%%%%%%%%%%%%%%%%%%%%%%%%%%%%%%%%%

	\section{Tripartite Hardy-type nonlocality}\label{hsioi}

	Let us consider three spatially separated observers: Alice, Bob, and Charlie. Each party independently performs randomly selected one of two possible local measurements, denoted as $A_x$, $B_y$ and $C_z$, respectively, where $x,y,z\in\{0,1\}$. Each measurement yields outcomes $a,b,c\in\{+1,-1\}$. The statistics obtained from such an experimental scenario are characterised by a vector, referred to as a behaviour $\vec{P} \in \mathbb{R}^{64}$, defined by $\vec{P}=\{p(a,b,c|x,y,z)\}$, where $p(a,b,c|x,y,z)$ represents the joint probability of outcomes. In quantum theory, this joint  probability is given by
	\begin{equation}
		p(a,b,c|x,y,z,\rho) = \Tr[\rho \ A_{a|x} \otimes B_{b|y} \otimes C_{c|z}]
	\end{equation}
	where $A_{a|x}$, $B_{b|y}$ and $C_{c|z}$ associated with outcomes $a$, $b$ and $c$, respectively. $\rho \in \mathscr{L}(\mathcal{H}_A\otimes \mathcal{H}_B \otimes \mathcal{H}_C)$ is the shared tripartite quantum state with $\mathcal{H}$ denotes the associated Hilbert space for each party's system.
	
	In the ontological model framework of quantum theory, preparing a quantum state $\rho$ corresponds to assigning the system an ontic state $\lambda \in \Lambda$, where $\Lambda$ denotes the ontic state space. If, in each experimental run, the outcomes for all parties are specified by $\lambda$, independently of others' measurement choices and outcomes, such a model is termed \textit{local ontic} model \cite{Bell1964}. Equivalently, local ontological models satisfy the assumptions of Settings Independence (SI) and Outcome Independence (OI) as shown in \cite{Jarrett1984, Shimony1993}.
	
	Unlike the bipartite case, the tripartite scenario involves much richer correlation structures, giving rise to several ontological models, including the fully-local (FL) model \cite{Mermin1990}, the Svetlichny bilocal (SvBL) model \cite{Svetlichny1987}, and the nosignalling bilocal (NSBL) model \cite{Bancal2013}. We present a general framework to classify such models based on the SI and OI assumptions.
	
	The assumption of two-way settings independence among all parties, denoted by $SI(3,3)$, is defined as follows
	\begin{equation} \label{si33}
		\begin{aligned}
			p(a|x,\lambda)&=p(a|x,y,z,\lambda);  \ \forall a,x,y,z,\lambda \\ 
			p(b|y,\lambda)&=p(b|x,y,z,\lambda); \ \forall b,x,y,z,\lambda \\  
			p(c|z,\lambda)&=p(c|x,y,z,\lambda); \ \forall c,x,y,z,\lambda
		\end{aligned}
	\end{equation}
	For example, $SI(2,2)$ denotes that any two parties (say Alice-Bob and Alice-Charlie) satisfy both way SI condition but not the other pair (Bob-Charlie).
	
	Similarly, two-way outcome independence, denoted by $OI(3,3)$, is defined as
	\begin{equation}\label{oi33}
		\begin{aligned}
			p(a|x,\lambda)&=p(a|x,b,c,\lambda);  \ \forall a,b,c,x,\lambda \\ 
			p(b|y,\lambda)&=p(b|y,a,c,\lambda); \ \forall a,b,c,y,\lambda \\  
			p(c|z,\lambda)&=p(c|z,a,b,\lambda); \ \forall a,b,c,z,\lambda
		\end{aligned}
	\end{equation}
	We show that under the assumptions $SI(3,3)\wedge OI(3,3)$, the joint probability distribution $p(a,b,c|x,y,z,\lambda)$ becomes fully factorisable, i.e., $p(a,b,c|x,y,z,\lambda)=p(a|x,\lambda)p(b|y\lambda)p(c|z\lambda) $ (see Appx.~\ref{apnA}). In the literature, such ontological models are referred to as FL-ontic models \cite{Mermin1990}.
	
	However, even if a joint probability distribution does not admit such a fully local decomposition, it may still be consistent with a weaker form of locality. For instance, under the assumptions $SI(3,3)\wedge OI(2,2)$, the distribution can be decomposed as $p(a,b,c|x,y,z,\lambda)=p(a|x,\lambda)p(b,c|y,z,\lambda)$ (see Appx.~\ref{apnA}). Ontological models satisfying this decomposition are referred to as NSBL-ontic models \cite{Gallego2012}. 
	
	Thus, if Alice, Bob, and Charlie's local events are space-like separated, i.e., they lie in mutually non-overlapping light-cones, then genuine tripartite nonlocality arises only when the correlations cannot be explained by a $SI(3,3)\wedge OI(2,2)$-ontic model.
	
	For a $SI(3,3)\wedge OI(2,2)$-ontic model to reproduce quantum mechanical statistics, the following condition must be satisfied
	\begin{eqnarray} \label{gnol}
		p(a,b,c|x,y,z,\rho)&=\sum\limits_{\lambda\in\Lambda} \Big[\mu_1(\lambda) \ p(a|x,\lambda)p(b,c|y,z,\lambda) \nonumber \\
		&+ \mu_2(\lambda) \ p(b|y,\lambda)p(a,c|x,z,\lambda)  \\
		&+\mu_3(\lambda) \ p(c|z,\lambda)p(a,b|x,y,\lambda) \Big] \nonumber
	\end{eqnarray}
	where $\mu_m(\lambda)$ are the probability distributions in the ontic state space, satisfying $0\leq\mu_m(\lambda)\leq 1 \ \ \forall m\in\{1,2,3\}$ and $\sum\limits_{\lambda\in\Lambda}\sum\limits_{m}\mu_m(\lambda) =1$.  Any quantum correlation that fails to admit the decomposition in Eq.~(\ref{gnol}) is said to exhibit nosignalling genuine nonlocality (NSGNL) \cite{Gallego2012, Bancal2013}. One may also consider a more relaxed model, namely, the $SI(2,2)\wedge OI(2,2)$-ontic model (SvBL ontic model), in which one pair of parties (e.g., Bob–Charlie) violates both the SI and OI. This leads to the possibility of signalling between them, which would imply that they lie within the same light-cone, thus not fulfilling the spatial separation assumption considered in our scenario.
	
	Beyond facet inequalities characterising these models, a Hardy-type argument has been proposed for the tripartite scenario \cite{Rahaman2014}, akin to Hardy’s original bipartite argument \cite{Hardy1992}. It is defined by the following constraints
	\begin{eqnarray}\label{Hardy3}
		&p_H= p(+1,+1,+1|0,0,0,\lambda) \label{eq:H3} \\
		&\begin{aligned} \label{tri}
			& p_{AB}(+1,+1 | 1,0,\lambda) = 0, \\
			&p_{BC}(+1,+1 | 1,0,\lambda) = 0, \\
			& p_{AC}(+1,+1 | 0,1,\lambda) = 0. \\
			& p(-1,-1,-1 |1,1,1,\lambda) = 0. \\
		\end{aligned}  
	\end{eqnarray} 
	While it has been shown that $p_H=0$ for all FL ontic models \cite{Rahaman2014}, we demonstrate that $p_H=0$ holds under the weaker condition $SI(3,3)\wedge OI(2,2)$ (see Appx.~\ref{apnA}). Thus, our result establishes that any nonzero values of $p_H$ certifies the failure of $SI(3,3)\wedge OI(2,2)$-ontic model and confirms genuine multipartite nonlocality. This means that Hardy violating $(p_H>0)$ correlations cannot be explained by any $SI(3,3)$-ontic model, even one permitting strong correlations (such as PR-box-type) between two of the parties, underscoring the fundamental strength of Hardy's nonlocality test in tripartite systems.
	
	Recent work also reveals that, even when all zero-probability constraints of Eqs.~(\ref{tri}) are met, quantum theory allows for $p_H>0$ for some tripartite entangled states and suitable measurements for Alice, Bob and Charlie \cite{Rahaman2014}. Such quantum correlations are called as genuine Hardy-nonlocal correlations.
	
	It has been shown \cite{Adhikary2024} that the quantum maximum value of $p_H$ is $p^{opt}_H=0.0181$, achieved using the following quantum state and measurements
	\begin{equation}
		\begin{aligned}
			\ket{\psi} &= c_0\ket{000}+c_1\Big[\ket{001}+\ket{010}+\ket{100}\Big]  \\
			& +c_2\Big[\ket{011}+\ket{101}+\ket{110}\Big]+c_3\ket{111},  \\
			A_0&=B_0=C_0=\sigma_z, \\
			A_1&=B_1=C_1=\qty(\alpha \beta^{*}+\alpha^{*} \beta)\sigma_x-\iota \qty(\alpha \beta^{*}-\alpha^{*} \beta) \sigma_y \\
			&+\qty(\lvert \alpha \rvert^2- \lvert \beta \rvert^2) \sigma_z 
		\end{aligned}
	\end{equation}
	where $c_0=\frac{|\alpha|^3|\beta|^3}{\sqrt{1-|\alpha|^6}}$, $c_1=\frac{-\beta |\alpha|^4|\beta|}{\sqrt{1-|\alpha|^6}}$, $c_2=\frac{\beta^2|\alpha|^5}{|\beta|\sqrt{1-|\alpha|^6}}$, and $c_3=\frac{\beta^3\sqrt{1-|\alpha|^6}}{|\beta|^3}$. $\alpha$ and $\beta$ are complex numbers satisfying $|\alpha|^2 =1-|\beta|^2=\frac{(17+3\sqrt{33})^{2/3}-(17+3\sqrt{33})^{1/3}-2}{3(17+3\sqrt{33})^{1/3}}$. 
	
	%%%%%%%%%%%%%%%%%%%%%%%%%%%%%%%%%%%%%%%%%%%%%%%%%%%%%%%%%%%%%%%%%%%%%%%%%%%%%%%%%%%%%%%%%%%%%%%%%%%%%%%%%%%%%%%%%%%%%%%%%%%%%%%%%%%%%%%%%%%%%%%%%%%%%%%%%%%%%%%%%%%%%%%%%%%%%%%%%%%%%%%%%%%%%%%%%%%%%%%%%%%%%%%%%%%%%%%%%%%%%%%%%%%%%%%%%%%%%%

	\section{Characterisation of the quantum behaviour showing Hardy's nonlocality} \label{hardybehav}
	
	Here we analyse the set of quantum behaviour exhibiting Hardy's nonlocality, represented (omitting the tensor product for simplicity) by
	\begin{equation}
		\vec{P} \equiv\qty{p(a,b,c|x,y,z,\rho)=\Tr[\rho A_{a|x}B_{b|y}C_{c|z}]} \in \mathcal{Q}
	\end{equation}
	where $ \mathcal{Q}$ denotes the set of all quantum correlations. By applying Naimark’s dilation theorem \cite{Paulsen2003}, the measurement operators $A_{a|x}$, $B_{b|y}$ and $C_{c|z}$ can, without loss of generality, be taken as projectors. Consequently, we denote these operators as $\Pi_{\cdot|\cdot}$ with $(\Pi_{\cdot|\cdot})^2=\Pi_{\cdot|\cdot}$. We begin with a general pure three-qubit state shared among Alice, Bob, and Charlie
	\begin{equation}\label{generalform}
		\begin{aligned}
			\ket{\psi}& = a_1 \ket{u_0v_0w_0}+a_2\ket{ u_1v_0w_0}+a_3\ket{u_0v_1w_0} \\
			&+ a_4\ket{u_0v_0w_1}+a_5\ket{u_1v_1w_0}+a_6\ket{u_0v_1w_1} \\
			&+ a_7\ket{u_1v_0w_1}+a_8\ket{ u_1v_1w_1}
		\end{aligned}
	\end{equation}
	where $a_i$ are complex coefficients satisfying $\sum_{i\in\{1,8\}}|a_{i}|^2=1$. The orthonormal basis vectors $\ket{u_x} \in \mathcal{H}_A^2$, $\ket{v_y}\in \mathcal{H}_B^2$ and $\ket{w_z}\in \mathcal{H}_C^2$ are expressed as follows
	\begin{equation}
		\begin{aligned}
			&\vert u_0\rangle=Co_\alpha \ket{0} + e^{\iota \phi}S_\alpha\ket{1},~\vert u_1\rangle=-S_\alpha\ket{0} + e^{\iota \phi}Co_\alpha\ket{1},\nonumber\\
			&\vert v_0\rangle=Co_\beta\ket{0} + e^{\iota \xi}S_\beta\ket{1},~\vert v_1\rangle=-S_\beta\ket{0} + e^{\iota \xi}Co_\beta\ket{1}, \nonumber\\
			&\vert w_0\rangle=Co_\gamma\ket{0} + e^{\iota \eta}S_\gamma\ket{1},~\vert w_1\rangle=-S_\gamma\ket{0} + e^{\iota \eta}Co_\gamma\ket{1}\nonumber
		\end{aligned}
	\end{equation}
	Here, $Co_k:=\cos{\frac{k}{2}},~S_k:=\sin{\frac{k}{2}}$, $\alpha,\beta, \gamma\in[0, \pi)$ and $\phi, \xi, \eta \in [0, 2\pi)$. The measurement observables for Alice, Bob, and Charlie are defined as
	\begin{eqnarray}
		A_0&=&\ket{0}\bra{0} -\ket{1}\bra{1},~~A_1=\ket{u_0}\bra{u_0}-\ket{u_1}\bra{u_1}, \nonumber \\
		B_0&=&\ket{0}\bra{0} -\ket{1}\bra{1},~~B_1=\ket{v_0}\bra{v_0}-\ket{v_1}\bra{v_1}, \label{eq2} \\
		C_0&=&\ket{0}\bra{0} -\ket{1}\bra{1},~~C_1=\ket{w_0}\bra{w_0}-\ket{w_1}\bra{w_1}, \nonumber
	\end{eqnarray}
	It is important to note here that without loss of generality, we fix $A_0=B_0=C_0=\sigma_z$ and keep $A_1$, $B_1$ and $C_1$ as well as the state $\ket{\psi}$ in their most general forms \cite{Rai2022}.
	
	Now, in order to satisfy the equality constraints on joint probabilities in the ontic space, given by Eq.~(\ref{tri}), and quantum theory as well, the state $\vert \psi\rangle$ must be orthogonal to $\ket{u_0 0 0}$, $\ket{u_0 0 1}$, $\ket{0 v_0 0}$, $\ket{1 v_0 0}$, $\ket{0 0 w_0}$, $\ket{0 1 w_0}$, $\ket{u_1 v_1 w_1}$. These seven orthogonality conditions allow the $15$ independent state parameters in $\ket{\psi}$ given by Eq.~(\ref{generalform}) to be expressed in terms of three measurement parameters $\alpha$, $\beta$ and $\gamma$ and their associated phases $\phi$, $\xi$ and $\eta$. The resulting class of pure three-qubit states, denoted by $\ket{\psi}_H$, satisfying tripartite Hardy conditions is given by
	\begin{eqnarray} \label{hardystate}
		\ket{\psi}_{H} &=&\frac{1}{\sqrt{N}} \Big[\ket{u_0v_0w_0} + \mathcal{C}_{\alpha}  \ket{u_1v_0w_0} + \mathcal{C}_{\beta}  \ket{u_0v_1w_0} \nonumber \\
		&+&\mathcal{C}_{\gamma}  \ket{u_0v_0w_1}+ \mathcal{C}_{\alpha} \mathcal{C}_{\beta} \ket{u_1v_1w_0} \\
		&+& \mathcal{C}_{\beta} \mathcal{C}_{\gamma} \ket{u_0v_1w_1}+\mathcal{C}_{\alpha} \mathcal{C}_{\gamma} \ket{u_1v_0w_1} \Big] \nonumber
	\end{eqnarray}
	where $N=1+\mathcal{C}^2_\alpha+\mathcal{C}^2_\beta+\mathcal{C}^2_\gamma+\mathcal{C}^2_\alpha \mathcal{C}^2_\beta+\mathcal{C}^2_\beta \mathcal{C}^2_\gamma+\mathcal{C}^2_\alpha \mathcal{C}^2_\gamma$ and $\mathcal{C}_k:=\cot\frac{k}{2}$ with $k \in \{\alpha, \beta, \gamma\}$. The quantum behaviour $\vec{P}_H \equiv \{p(a,b,c|x,y,z,\rho_H)\}$, represented by 64 joint probabilities derived from the state $\ket{\psi}_H$ in Eq.~(\ref{hardystate}) and measurements, given by Eq.~\eqref{eq2} can be expressed in an array as follows
	\begin{widetext}
		{\scriptsize \begin{align}
				\begingroup
				\setlength{\tabcolsep}{20pt} 
				\renewcommand{\arraystretch}{2}
				\begin{array}{c||c|c|c|c|c|c|c|c|} 
					&  ( +,+,+ )  &  ( -,+,+ ) & ( +,-,+ ) & ( +,+,- ) &( -,-,+ )  &  ( -,+,- ) & ( +,-,- ) & ( -,-,- ) \\ \hline\hline
					A_0B_0C_0   & \dfrac{( 1 - r ) r^2 s t h }{(1 - r s)(1 - r t)}  & \dfrac{( 1 - r )r t h}{1 - r t}  & \dfrac{( 1 - r )r s h}{1 - r s} & \dfrac{( 1 - r )^2 r s t}{(1 - r s)(1 - r t)} & ( 1 - r )h  & \dfrac{( 1 - r )^2 t}{ 1 - r t } & \dfrac{( 1 - r )^2 s}{ 1 - r s } &  r  \\\hline
					A_1B_0C_0   &  0  &  \dfrac{( 1 - r ) r t h }{(1 - r s)(1 - r t)}  &  0  &  0 &  \dfrac{( 1 - r ) h}{1 - r s}  &  \dfrac{( 1 - r )^2 t}{(1 - r s)(1 - r t)}  &  s  &  \dfrac{r ( 1 - s )^2}{ 1 - r s }   \\\hline
					A_0B_1C_0   &  0  &  0  & \dfrac{( 1 - r ) r s h }{(1 - r s)(1 - r t)}  & 0 &  \dfrac{( 1 - r ) h}{1 - r t}  &  t  & \dfrac{( 1 - r )^2 s}{(1 - r s)(1 - r t)}  & \dfrac{r ( 1 - t )^2}{ 1 - r t }  \\\hline
					A_1B_1C_0   &  0  &  0  & 0 &  r s t &  \dfrac{( 1 - r ) h}{(1 - r s)(1 - r t)}  &  (1 - r s) t  & s (1 - r t) &  \dfrac{rh^2}{(1 - r s)(1 - r t)}  \\\hline
					A_0B_0C_1   & 0  & 0  & 0 & (1-r) r s t  & \dfrac{h}{(1 - r s)(1 - r t)}  & (1 - r) (1 - r s) t   & (1 - r) s (1 - r t) & \dfrac{r(1-r)(1-h)^2}{(1 - r s)(1 - r t)}   \\\hline
					A_1B_0C_1   &  0  &  0  &  \dfrac{r s h }{(1 - r s)(1 - r t)}  &  0 &  \dfrac{h}{1 - r t}  &  (1 - r) t  &  \dfrac{(1-r) s}{(1 - r s)(1 - r t)}  &  \dfrac{(1 - r) r s^2_2}{1 - r t}  \\\hline
					A_0B_1C_1   &  0  &  \dfrac{r t h }{(1 - r s)(1 - r t)}  & 0  & 0 &  \dfrac{h}{1 - r s}  &  \dfrac{(1-r) t}{(1 - r s)(1 - r t)}  & ( 1 - r ) s  & \dfrac{(1 - r) r s^2_1}{1 - r s}  \\\hline
					A_1B_1C_1   &  \dfrac{r^2 s t h }{(1 - r s)(1 - r t)}  &  \dfrac{r t h }{1 - r t}  & \dfrac{r s h }{1 - r s} &  \dfrac{( 1 - r ) r s t}{(1 - r s)(1 - r t)} &  h  &  \dfrac{(1-r) t}{1 - r t}  & \dfrac{( 1 - r ) s}{1 - r s} &  0  \\\hline
				\end{array}
				\endgroup
				\label{hardybehavtable}
		\end{align}}
	\end{widetext}
	where $h:=(1-s-t+r s t)$, $r:=1-Co^2_\alpha Co^2_\beta Co^2_\gamma$, $s:=r^{-1}S^2_\alpha$ and $t:=r^{-1}S^2_\beta$. The behaviour $\vec{P}_{H}$ is nonlocal if and only if $p(+1,+1,+1|A_0,B_0,C_0)>0$, implying $r \in (0,1)$, $s \in (0,1)$ and $t \in \qty(0,\frac{1-s}{1-r s})$. Values of $\alpha$, $\beta$ and $\gamma$ can be obtained from $\alpha=2\sin^{-1} \sqrt{rs}$, $\beta=2\sin^{-1} \sqrt{rt}$, $\gamma=2\sin^{-1} \sqrt{\frac{rh}{(1 - r s)(1 - r t)}}$. The optimal quantum value of the Hardy probability is $p_H^{opt}=0.0181$ for $r=0.8392$ and $s=t=0.5436$ \cite{Adhikary2024}.
	
	Now, the three-qubit state in Eq.~(\ref{hardystate}) when expressed in the computational basis $\{\ket{0},\ket{1}\}$, and in terms of the parameters $r$, $s$ and $t$, reduces to the following form
	\begin{widetext}
		\begin{eqnarray} \label{hardystate}
			\ket{\psi(r,s,t)}_{H}&=&\sqrt{\frac{(1-r) r^2 s t h }{(1-r s)(1-r t)}}~\ket{000} - e^{\iota \phi}\sqrt{\frac{(1-r)r t h}{1-r t}}~\ket{100} \nonumber \\
			&-& e^{\iota \xi}\sqrt{\frac{(1-r)r s h}{1-r s}}~\ket{010}-e^{\iota \eta}\sqrt{\frac{(1-r)^2 r s t}{(1-r s)(1-r t)}}~\ket{001}+e^{\iota (\xi+\eta)}\sqrt{\frac{(1-r)^2 s}{1-r s}}~\ket{011}\nonumber \\
			&+&e^{\iota (\phi+\xi)}\sqrt{(1-r)h}~\ket{110}+e^{\iota (\phi+\eta)}\sqrt{\frac{(1-r)^2 t}{1-r t}}~\ket{ 101}+e^{\iota (\phi+\xi+\eta)}\sqrt{r}~\ket{111}    
		\end{eqnarray}
	\end{widetext}
	
	%%%%%%%%%%%%%%%%%%%%%%%%%%%%%%%%%%%%%%%%%%%%%%%%%%%%%%%%%%%%%%%%%%%%%%%%%%%%%%%%%%%%%%%%%%%%%%%%%%%%%%%%%%%%%%%%%%%%%%%%%%%%%%%%%%%%%%%%%%%%%%%%%%%%%%%%%%%%%%%%%%%%%%%%%%%%%%%%%%%%%%%%%%%%%%%%%%%%%%%%%%%%%%%%%%%%%%%%%%%%%%%%%%%%%%%%%%%%%%%%%%%%%%%%%%%%%%%%%%%%%%%%%%%%%%%%%%%%%%%%%%%%%%%%%%%%%%%%%%%%%%%%%%%%%%%%%%%%%%%%%%%%%%%%%%%%%%%%%%%%%%%%%%%%%%%%%%%%%%%%%%%%%%%%%%%%%%%%%%%%%%%%%%%%%%%%%%%%%%%%%%%%%%%%%%%%%%%%%%%%%%%%%%%%%%%%%%%%%%%%%%%%%%%%%%%%%%%%%%%%%%%%%%%%%%%%%%%%%%%%
	
	\section{Self-testing with nonmaximal violation of tripartite Hardy's nonlocality} \label{stth}
	
	Self-testing a quantum state involves performing a black-box experiment under the assumption of independent and identically distributed operations, where the quantum state and the measurements within the black box are unknown to Alice, Bob, and Charlie. This lack of knowledge about the dimensionality of the underlying Hilbert space prevents the parametrisation of the measurements and states employed in the experiment, unlike the approach used in Sec.~\ref{hardybehav} to evaluate the Hardy nonlocal behaviour described by Eq.(\ref{hardybehavtable}). The earlier evaluation of Eq.(\ref{hardybehavtable}) assumed that each local Hilbert space dimension is $2$ and that the measurements are projective. However, the observed behaviour could result from more general measurement operators of arbitrary dimension, specifically positive operator-valued measures (POVMs), applied to the shared quantum state rather than from projective measurements alone.
	
	Given that the dimensions of the state space and the observables are finite but unspecified, Neumark's dilation theorem enables us to confine the analysis to projective measurements \cite{Paulsen2003}. This simplification ensures that, without loss of generality, the dichotomic observables of Alice, Bob, and Charlie, acting on $\mathcal{H}^d$, can always be represented as follows
	\begin{equation}
		\mathbf{A}_x \equiv \{\mathbf{A}_{a|x}\}; \  \mathbf{B}_y \equiv \{\mathbf{B}_{b|y}\}; \  \mathbf{C}_z \equiv \{\mathbf{C}_{c|z}\}
	\end{equation}
	where $a,b,c\in\{\pm 1\}$ and $M_{\pm|\cdot}$ are projectors corresponding to the observables $M_{\cdot}$, satisfying $M_{\pm|\cdot}\geq 0$, $(M_{\pm|\cdot})^2=\openone$. The probability distribution observed from the black-box device is then expressed as
	\begin{equation}
		\mathfrak{p}(a,b,c|x,y,z,\boldsymbol{\rho})=\Tr[\boldsymbol{\rho}\mathbf{A}_{a|i}\mathbf{B}_{b|j} \mathbf{C}_{c|k}]
	\end{equation}
	where $\boldsymbol{\rho} \in \mathscr{L}(\mathcal{H}^d_A\otimes \mathcal{H}^d_B\otimes \mathcal{H}^d_C)$ be an arbitrary three-qudit state. We denote the observed behaviour as $\vec{\mathfrak{P}}\equiv\{\mathfrak{p}(a,b,c|x,y,z,\boldsymbol\rho)\}$ and bold-symbols are used to represent hermitian operators in arbitrary dimensional Hilbert space.
	
	Now, the Jordan lemma asserts that, for each party, there exists an orthonormal basis for $\mathcal{H}^d$ in which all four projectors $\{\mathbf{M}_{\pm|\cdot}\}$, corresponding to the two measurement choices per party, are simultaneously block diagonal, with a maximum block size of $2\times 2$ \cite{Masanes2006}. This orthonormal basis induces a direct sum decomposition, $\mathcal{H}=\oplus_m\mathcal{H}^m$, where each subspace $\mathcal{H}^m$ has a dimension of at most two. Consequently, the four projectors associated with each party decompose as $\mathbf{M}_{\pm|\cdot}=\oplus_m M_{\pm|\cdot}^m$, where each component $M_{\pm|\cdot}^m$ acts on the subspace $\mathcal{H}^m$. This implies that the dichotomic observables of Alice, Bob, and Charlie, as well as the shared three-qudit state, can be expressed as
	\begin{equation}
		\begin{aligned}
			&\boldsymbol\rho=\bigoplus_{m} \mu_{m} \rho_{m} \\
			&\mathbf{A}_x=\bigoplus_m A_x^{m}; \ \mathbf{B}_y=\bigoplus_m B_y^{m}; \ \mathbf{C}_z=\bigoplus_m C_z^{m} 
		\end{aligned}
	\end{equation}
	where $A_x^{m},B_y^{m},C_z^{m} \in \mathscr{L}(\mathcal{H}^2)$ and $\rho_{m}\in \mathscr{L}(\mathcal{H}^2_A\otimes \mathcal{H}^2_B\otimes \mathcal{H}^2_C)$. Each qubit measurement acting on $\rho_{m}$ produces a set of joint probability distributions, $\vec{P}_m$, given by
	\begin{equation}
		\vec{P}_m= \qty{p_m(a,b,c|x,y,z,\rho_{m})=\Tr[\rho_{m} A_{a|x}^{m} B_{b|y}^{m} C_{c|z}^{m}]}
	\end{equation}
	We refer to each such decomposition as a particular strategy, $S_m$, for producing the qubit-realised behaviour $\vec{P}_m$. Using the Jordan lemma, the observed probability $\mathfrak{p}(a,b,c|x,y,z,\boldsymbol{\rho})$ can be represented as a convex mixture of joint probabilities arising from each strategy $S_m$, as follows 
	\begin{equation}\label{jlprob}
		\begin{aligned}
			\mathfrak{p}(a,b,c|x,y,z,\boldsymbol\rho)&=\sum\limits_{m} \mu_{m} p_m(a,b,c|x,y,z,\rho_{m}) \\
			&=\sum\limits_{m} \mu_{m} \Tr[\rho_{m} A_{a|x}^{m} B_{b|y}^{m} C_{c|z}^{m}]
		\end{aligned}
	\end{equation}
	
	Note that if the observed Hardy probability, $\mathfrak{p}_H=\mathfrak{p}(+1,+1,+1|000,\boldsymbol\rho)$, attains its maximum value, it necessarily follows that all the Hardy probabilities, $p^m_H=p_m(+1,+1,+1|0,0,0,\rho_{m})$ arising from the qubit strategies must also reach their maximum value. From Eq.~(\ref{hardybehavtable}), it uniquely determines the values of $r$, $s$ and $t$, thereby fixing $\rho_{m}=\rho$ according to Eq.~(\ref{hardystate}). Thus, the maximum value of $\mathfrak{p}_H$ enables the self-testing of the pure nonmaximal three-qubit state $\rho$ using Jordan's lemma \cite{Rabelo2012, Rai2021, Adhikary2024,Adhikary2024pla}.
	
	However, if the observed Hardy-probability does not achieve its optimal value, such that $\mathfrak{p}_H=\delta<\mathfrak{p}_H^{opt}$, it no longer ensures that each individual three-qubit Hardy probability satisfies $p^m_H=\delta$. This is because there may exist some three-qubit states and qubit measurements for which $p^m_H \lesseqgtr \delta$, while the weights $\mu_{m}$ are adjusted in such a way that the overall probability becomes $\mathfrak{p}_H=\delta$. Consequently, relying solely on the three-qubit behaviour described by Eq.~(\ref{hardybehavtable}) does not guarantee a unique state and measurements to achieve such a Hardy-probability. Therefore, self-testing of three-qubit states from a nonmaximal violation of Hardy probability is not straightforward.
	
	Nevertheless, this challenge serves as the basis for the primary result of this paper, demonstrated through the following proposition.
	
	\begin{prop}
		If a black-box experiment records the following probabilities
		\begin{equation} \label{prop1rst}
			\begin{aligned}
				\mathfrak{p}(-1,-1,-1| i=0,j=0,k=0,\boldsymbol{\rho})&=\mathbf{r} \\
				\mathfrak{p}(+1,-1,-1| i=1,j=0,k=0,\boldsymbol{\rho})&=\mathbf{s} \\
				\mathfrak{p}(-1,+1,-1|i=0,j=1,k=0,\boldsymbol{\rho})&=\mathbf{t}
			\end{aligned}
		\end{equation}
		and if the remaining probabilities can be expressed in terms of $\mathbf{r}$, $\mathbf{s}$ and $\mathbf{t}$ in a manner consistent with the form of Eq.~(\ref{hardybehavtable}), then there exists a set of values for $r$, $s$ and $t$ such that the state of the unknown system is uniquely determined by these values. Specifically, the state satisfies $\boldsymbol{\rho}(\mathbf{r},\mathbf{s},\mathbf{t})=\bigoplus_{i} \mu_{i}\rho (r,s,t)$, where $\rho (r,s,t)=\ket{\psi(r,s,t)}_H\bra{\psi(r,s,t)}_H$. This means that $\boldsymbol{\rho}(\mathbf{r},\mathbf{s},\mathbf{t})$ is equivalent, up to local isometries, to $ \zeta \otimes\ket{\psi(r,s,t)}_H\bra{\psi(r,s,t)}_H$. Here, $\ket{\psi(r,s,t)}_H$ is the pure three-qubit nonmaximally entangled state defined by Eq.~(\ref{hardystate}) and, $\zeta$ is a arbitrary tripartite state.
	\end{prop}
	
	The central idea underpinning the proof of our claim is derived from the Jordan lemma, Neumark's dilation theorem, and the concept of a concave cover. 
	
	In a black-box experiment, let the observed Hardy probability be $\mathfrak{p}_H=\delta$, where $\delta\in(0,0.0181)$, and let the observed values of $\mathbf{r}$, $\mathbf{s}$ and $\mathbf{t}$ be determined by the corresponding joint probabilities given by Eq.~(\ref{prop1rst}). Denote the set of behaviours arising from qubit strategies as $\{\vec{P}_m\}$. Generally, these behaviours can be classified into two types - (i) $\{\vec{P}'_m | p^m_H\lesseqgtr \delta \}$, and (ii) $\{\vec{P}^{*}_m | p^m_H = \delta \}$. 
	
	In the following, we demonstrate the existence of a set of 3-tuples $(r_i,s_j,t_k)$ for which the observed behaviour admits a unique decomposition of the form $\vec{\mathfrak{P}}=\sum_{ijk} \mu_{ijk} \vec{P}^{*}_{ijk}$ for each value of $\delta$. Such decompositions are the extremal points of the convex hull of the set of points $\{r_i,s_j,t_k,\Omega(r_i,s_j,t_k)\}$, where 
	\begin{equation}
		\Omega(r_i,s_j,t_k)=\frac{r_i^2(1 - r_i)s_jt_k(1 - s_j- t_k+ r_is_jt_k)}{(1 - r_i s_j)(1 - r_i t_k)}
	\end{equation}
	Thus, the goal is to characterise the points $(r_i,s_j,t_k)$ where the function $\Omega(r_i,s_j,t_k)$ produces the same value $\delta$ for all decompositions. Specifically, we aim to identify the conditions under which we can conclude that $r_{i}=\mathbf{r}$; $\ s_{j}=\mathbf{s}$; $\ t_{k}=\mathbf{t}$. By extending the reasoning applied to the maximal Hardy probability, each Hardy value will correspond to a unique set of 3-tuples, such that the Hardy probabilities derived from the qubit strategies match the observed value. This establishes the self-testing of a broad class of pure nonmaximally entangled three-qubit states.
	
	Now, let us define a function of $(r_i,s_j,t_k)$ using the probabilities in Eq.~(\ref{hardybehavtable}) as
	\begin{equation}
		\Omega(\mathbf{r},\mathbf{s},\mathbf{t})=\sum_{a,b,c,x,y,z} c_{abcxyz} ~p(a,b,c|x,y,z,\rho)+c_0
	\end{equation}
	where $c_{abcxyz}$ and $c_0$ are some real coefficients. By applying Eq.~(\ref{jlprob}) to each probability term, we obtain
	\begin{equation}\label{eq8}
		\Omega (\mathbf{r},\mathbf{s},\mathbf{t})= \sum_{i,j,k} \mu_{ijk} \Omega (r_{i}, s_{j}, t_{k})
	\end{equation}
	where the values of $r_i$, $s_j$ and $t_k$ are given by
	\begin{equation}
		\begin{aligned}
			r_{i}&=p(-1,-1,-1| 0,0,0,\rho) \\
			s_{j}&=p(+1,-1,-1|1,0,0,\rho) \\
			t_{k}&=p(-1,+1,-1|0,1,0,\rho)
		\end{aligned}
	\end{equation}
	Furthermore, when the black-box statistics satisfies the zero constraints of Hardy’s test, the same holds for every subspace.
	
	To formalise our result, we extend Lemma 1 from \cite{Rai2022} to our considered scenario.
	
	\begin{lem}\label{Lemma1}
		Let {\small$\mathcal{E}(\mathbf{r},\mathbf{s},\mathbf{t}):(0,1)^3 \rightarrow  \mathbb{R}$} be the concave cover of {\small$\Omega (\mathbf{r},\mathbf{s},\mathbf{t})$}, and suppose $\mathfrak{D}$ is the set of points from the domain for which {\small$\mathcal{E} (\mathbf{r},\mathbf{s},\mathbf{t})=\Omega (\mathbf{r},\mathbf{s},\mathbf{t})$}. Then $\forall \mathbf{r},\mathbf{s},\mathbf{t}\in \mathfrak{D}$, and $\forall i,j,k$
		\begin{equation}
			r_{i}=\mathbf{r}; \ s_{j}=\mathbf{s}; \ t_{k}=\mathbf{t}
		\end{equation}
		provided that $\Omega(\mathbf{r},\mathbf{s},\mathbf{t})$ is a strictly concave function of $\mathbf{r},\mathbf{s},\mathbf{t}$ over the region $\mathfrak{D}$. 
	\end{lem}
	
	\begin{proof}
		A \textit{concave cover} $\mathcal{E}(\mathbf{r},\mathbf{s},\mathbf{t})$, for the function $\Omega(\mathbf{r},\mathbf{s},\mathbf{t})$ over its domain $D\in(0,1)^3$, is defined as the smallest concave overestimator of $\Omega(\mathbf{r},\mathbf{s},\mathbf{t})$ over $D$. Formally, for a function $\Omega:D\to \mathbb{R}$, the concave cover $\mathcal{E}(\mathbf{r},\mathbf{s},\mathbf{t})$ is
		\begin{equation}\label{concov}
			\mathcal{E}(\mathbf{r},\mathbf{s},\mathbf{t})=\sup \qty{\omega(\mathbf{r},\mathbf{s},\mathbf{t})|\omega\geq\Omega(\mathbf{r},\mathbf{s},\mathbf{t}) \ \text{on} \ D}
		\end{equation}
		
		Applying Jensen's inequality\footnote{Jensen's inequality \cite{Durrett2019book}: For any concave real-valued function $f(x):\mathbb{R}^n\to \mathbb{R}$, the following inequality holds
			\begin{equation}\label{ji}
				f\qty(\sum_{m=1}^n \mu_m x_m) \geq \sum_{m=1}^n \mu_m \ f\qty(x_m)
			\end{equation}
			where $x_m\in\mathbb{R}^n$ and $\mu_m\geq 0$ for all $m$, with $\sum_m \mu_m=1$. Equality is achieved \textit{iff} $x_1=x_2=\ldots=x_n$, provided that $f(x)$ is nonlinear.}  to the function $\mathcal{E}(\mathbf{r},\mathbf{s},\mathbf{t}):\mathbb{R}^3\rightarrow \mathbb{R}$, we obtain
		\begin{equation}\label{eq9}
			\mathcal{E}\qty(\sum_{i,j,k} \ \mu_{ijk} (r_{i}, s_{j}, t_{k}))
			\geq \sum_{ijk} \mu_{ijk} \ \mathcal{E}\qty(r_{i},s_{j}, t_{k})
		\end{equation}
		Moreover in the domain when $\Omega$ is strictly concave, we have $\mathcal{E}(\mathbf{r},\mathbf{s},\mathbf{t}) =\Omega (\mathbf{r},\mathbf{s},\mathbf{t})$. Since, by the definition of concave cover, $ \Omega(r_{i}, s_{j}, t_{k})\leq\mathcal{E}(r_{i}, s_{j}, t_{k})$, using Eq.~\eqref{eq8}, we get
		\begin{equation}\label{eq10}
			\mathcal{E}(\mathbf{r},\mathbf{s},\mathbf{t}) =\sum_{i,j,k} \mu_{ijk} \ \Omega(r_i, s_j, t_k) 
			\leq \sum_{i,j,k} \mu_{ijk}~\mathcal{E}(r_i, s_j, t_k).
		\end{equation}
		Thus, Eqs.~(\ref{eq9}) and Eq.~(\ref{eq10}) together imply that in the domain $(\mathbf{r},\mathbf{s},\mathbf{t})\in \mathfrak{D}$
		\begin{equation}\label{eq11}
			\mathcal{E}\qty(\sum_{i,j,k} \mu_{ijk}( r_i,s_j,t_k)) = \sum_{ijl} \mu_{ijk} \ \mathcal{E}\qty(r_i, s_j, t_k). 
		\end{equation}
		Therefore, if $\Omega(\mathbf{r},\mathbf{s},\mathbf{t})$ is a strictly concave function of $\mathbf{r}$, $\mathbf{s}$ and $\mathbf{t}$ over the region $\mathfrak{D}$, then in every subspace $\mathcal{H}_A^i\otimes \mathcal{H}_B^j\otimes \mathcal{H}_C^k$, the parameter values  $(r_i,s_j,t_k)$ must be identical, i.e., $r_i=\mathbf{r}$, $s_j=\mathbf{s}$ and $t_k=\mathbf{t}$ for all $i,j,k$. 
	\end{proof}
	
	Let us now analyse an example where Lemma~\ref{Lemma1} is applicable. We consider a function defined by the Hardy success probability
	\begin{equation}\label{eq12}
		\begin{aligned}
			\Omega^{\ast}(\mathbf{r},\mathbf{s},\mathbf{t}) & :=  \mathfrak{p}(+1,+1,+1|0,0,0,\boldsymbol{\rho}) \\
			&=\frac{\mathbf{r}^2(1 - \mathbf{r})\mathbf{s}\mathbf{t}(1 - \mathbf{s}- \mathbf{t}+ \mathbf{r}\mathbf{s}\mathbf{t})}{(1 - \mathbf{r} \mathbf{s})(1 - \mathbf{r} \mathbf{t})}.
		\end{aligned}
	\end{equation}
	The function $\Omega^{\ast}(\mathbf{r},\mathbf{s},\mathbf{t})$ is nonlinear. Note that the concave cover of a function is the function itself if the function is strictly concave. Otherwise, one must define an approximate function which overestimates the function, and such a construction for any general function is challenging \cite{Tawarmalani2013, Benameur2017}. To verify concavity of $\Omega^{*}$, we compute its Hessian matrix\footnote{For a three-variable function, the hessian matrix is given by: \[
		\mathbf{H}_{\Omega^{*}}(\mathbf{r}, \mathbf{s}, \mathbf{t}) = \begin{bmatrix}
			\frac{\partial^2 \Omega^{*}}{\partial \mathbf{r}^2} & \frac{\partial^2 \Omega^{*}}{\partial \mathbf{r} \partial \mathbf{s}} & \frac{\partial^2 \Omega^{*}}{\partial \mathbf{r} \partial \mathbf{t}} \\
			\frac{\partial^2 \Omega^{*}}{\partial \mathbf{s} \partial \mathbf{r}} & \frac{\partial^2 \Omega^{*}}{\partial \mathbf{s}^2} & \frac{\partial^2 \Omega^{*}}{\partial \mathbf{s} \partial \mathbf{t}} \\
			\frac{\partial^2 \Omega^{*}}{\partial \mathbf{t} \partial \mathbf{r}} & \frac{\partial^2 \Omega^{*}}{\partial \mathbf{t} \partial \mathbf{s}} & \frac{\partial^2 \Omega^{*}}{\partial \mathbf{t}^2}
		\end{bmatrix}
		\]
	} \cite{Boyd2009book}. In particular, a function is strictly concave if all the eigenvalues of the Hessian matrix is negative everywhere on $D$. Conversely, the function is strictly convex if all the eigenvalues are positive across $D$. If the Hessian exhibits both positive and negative eigenvalues, the function is deemed neither convex nor concave over the domain. 
	
	To evaluate these properties computationally, we discretise the domain $D$ into a fine grid of points. At each grid point $(\mathbf{r},\mathbf{s},\mathbf{t})$,  the second-order partial derivatives of the function $\Omega^{*}$ are calculated symbolically using SymPy, a Python library for symbolic mathematics \cite{sympy}. These derivatives are then assembled into the Hessian matrix comprising all second-order mixed partial derivatives. Using NumPy \cite{Harris2020}, the eigenvalues of the Hessian at each grid point are numerically computed and analysed to determine the definiteness of the matrix.
	
	This numerical analysis conducted using Python reveals that $\Omega^{*}$ is strictly concave only in a subdomain $\mathfrak{D}^{*} \subset D$, while elsewhere it is indefinite, i.e., neither concave nor convex (see Fig.~\ref{fig:Hess}).
	
	\begin{figure}[htpb!]
		\centering
		\includegraphics[width=\linewidth]{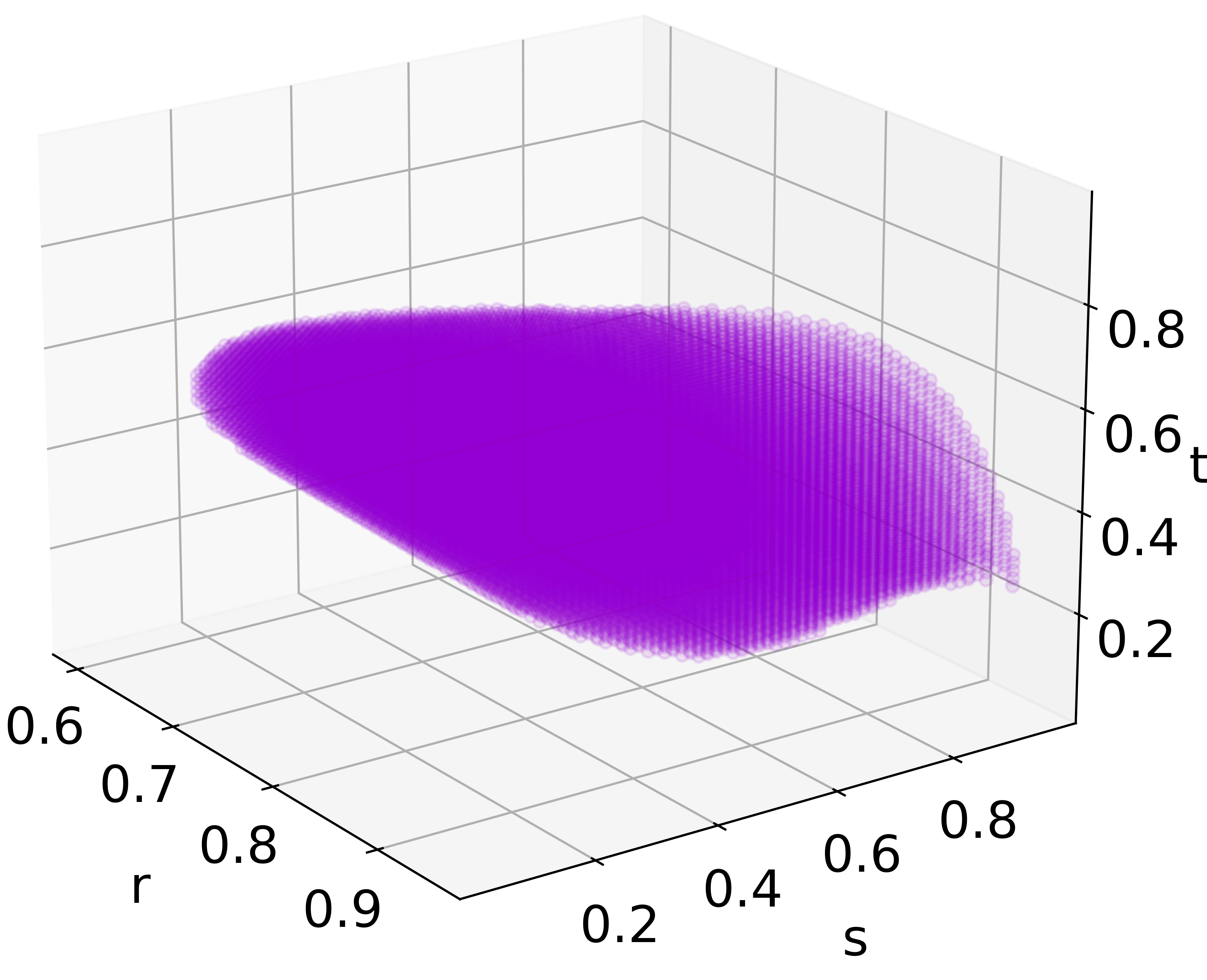}
		\caption{The violet curve lies within the three-dimensional region in $(\mathbf{r},\mathbf{s},\mathbf{t})\in(0,1)^3$ space, representing the subset where all the eigenvalues of the Hessian is negative. In the remaining region, all the eigenvalues are not negative. This implies that $\Omega^{*}$ is strictly concave only over the violet-coloured subspace of $(\mathbf{r},\mathbf{s},\mathbf{t})$.}
		\label{fig:Hess}
	\end{figure}
	This partitioning of the domain is critical in the construction of the concave envelope, which, as discussed later, is only defined via convex hull methods over regions where the function fails to be strictly concave.
	
	Since the function is not strictly concave over its entire domain, the analytical construction of the concave cover function and the characterisation of the corresponding region is a challenging problem to solve. Therefore, we take recourse to a numerical approach using Python to construct the concave cover.
	
	In the following, we show that for this function, there exists a concave cover $\mathcal{E}^{\ast}(\mathbf{r},\mathbf{s},\mathbf{t})$ and a region $\mathfrak{D}^{\ast}\subset (0,1)\times(0,1)\times(0,1)$ such that $\mathcal{E}^{\ast}(\mathbf{r},\mathbf{s},\mathbf{t})=\Omega^{\ast}(\mathbf{r},\mathbf{s},\mathbf{t})$. 
	
	From the definition of the concave cover given by Eq.~(\ref{concov}), the supremum can be viewed geometrically as the upper boundary of the convex hull of the hypograph of the function $\Omega^{*}$ over $D$.  The hypograph of $\Omega^{*}$ is given by \cite{Boyd2009book}
	\begin{equation}
		hypo(\Omega^{*})=\qty{(\mathbf{r},\mathbf{s},\mathbf{t},\omega(\mathbf{r},\mathbf{s},\mathbf{t}))\in D \times \mathbb{R}|\omega \leq \Omega^{*}}    
	\end{equation}
	By discretising $D$ into a grid of points and evaluating the corresponding values of $\Omega^{*}$ for each grid point, we compute the convex hull of $hypo(\Omega^{*})$. This is evaluated using Python's \textit{scipy.spatial.ConvexHull} implementation \cite{Virtanen2020}, which interfaces with the QHull algorithm via compiled C libraries \cite{Barber1996}.
	
	QHull is a robust algorithm for computing the convex hull of a set of points in $n$-dimensional space. In two dimensions $(n=2)$,  it begins by identifying the leftmost and rightmost points to form a baseline, then recursively selects the point farthest from this line. The algorithm divides the remaining points into subsets and continues this process, progressively constructing the convex hull. In higher dimensions, QHull generalises this approach by constructing the convex hull as a polytope composed of facets. It begins with an initial simplex (e.g., a line segment in $2D$, a triangle in $3D$, or a tetrahedron in $4D$) and incrementally adds new points. For each new point, the algorithm identifies which facets are visible from it, removes these facets, and creates new ones by connecting the point to the boundary of the visible region. The extremal points of the resulting convex hull—accessible via \textit{hull.vertices}, define the concave cover of $\Omega^{*}$ over $D$. The complete Python implementation for this procedure is detailed in \cite{gitadhikary}.
	
	It is important to note that the concave cover of $\Omega^{*}$ can only be computed over a convex domain $D$. However, not all regions within the domain yield valid Hardy probability values; in particular, there may be subregions where $\Omega^{*}\leq 0$ (see Fig.~\ref{fig:chulldia}). Thus, after constructing the concave cover, one must filter out points for which $\Omega^{*}\leq 0$, thereby isolating the relevant subregion where $\Omega^{*}>0$. This subregion constitutes the effective concave sub-cover and is illustrated in Fig.~\ref{fig:chulldia}.
	
	\begin{figure*}[t]
		\centering
		\includegraphics[width=0.32\textwidth]{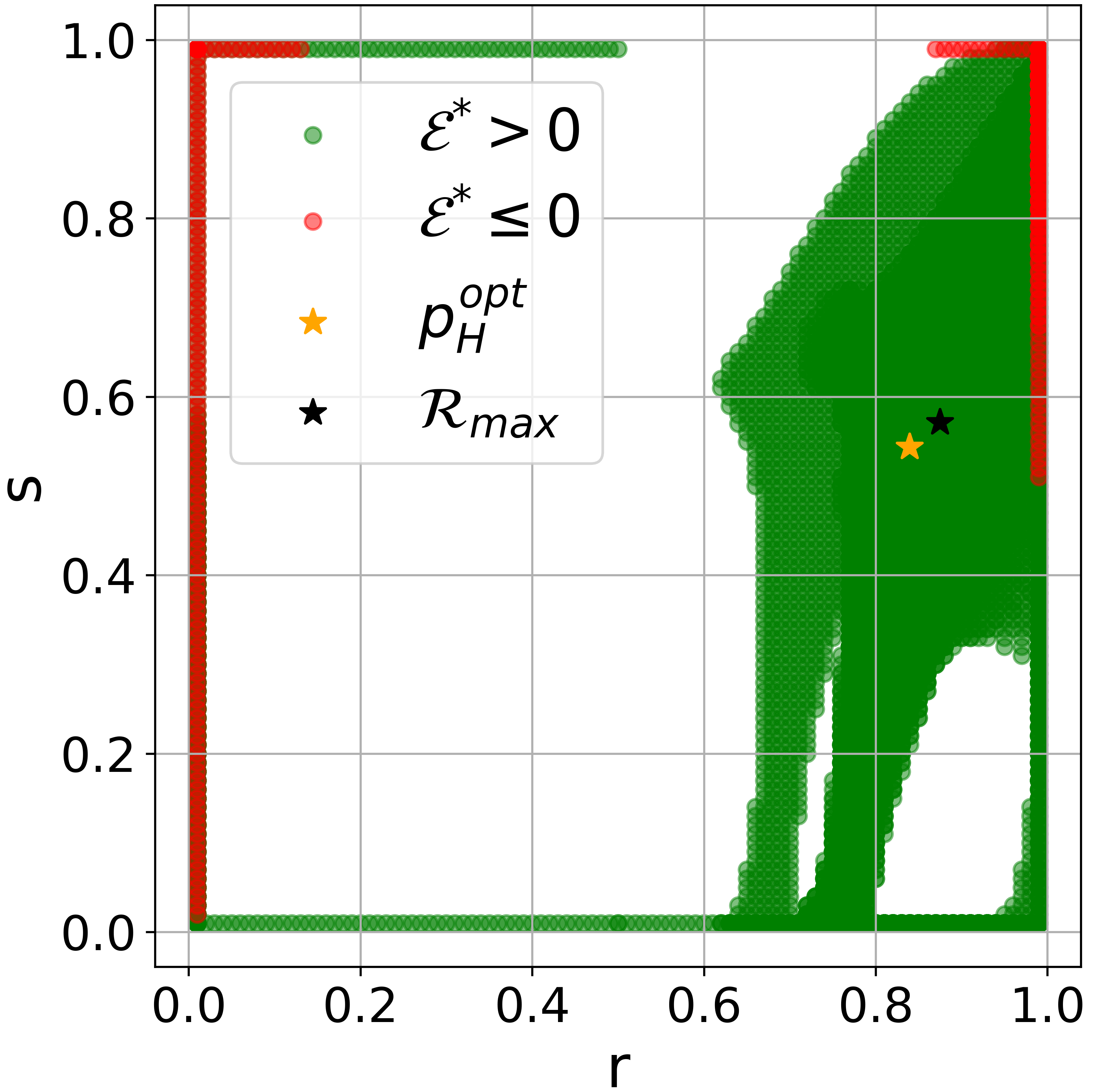}
		\hfill
		\includegraphics[width=0.32\textwidth]{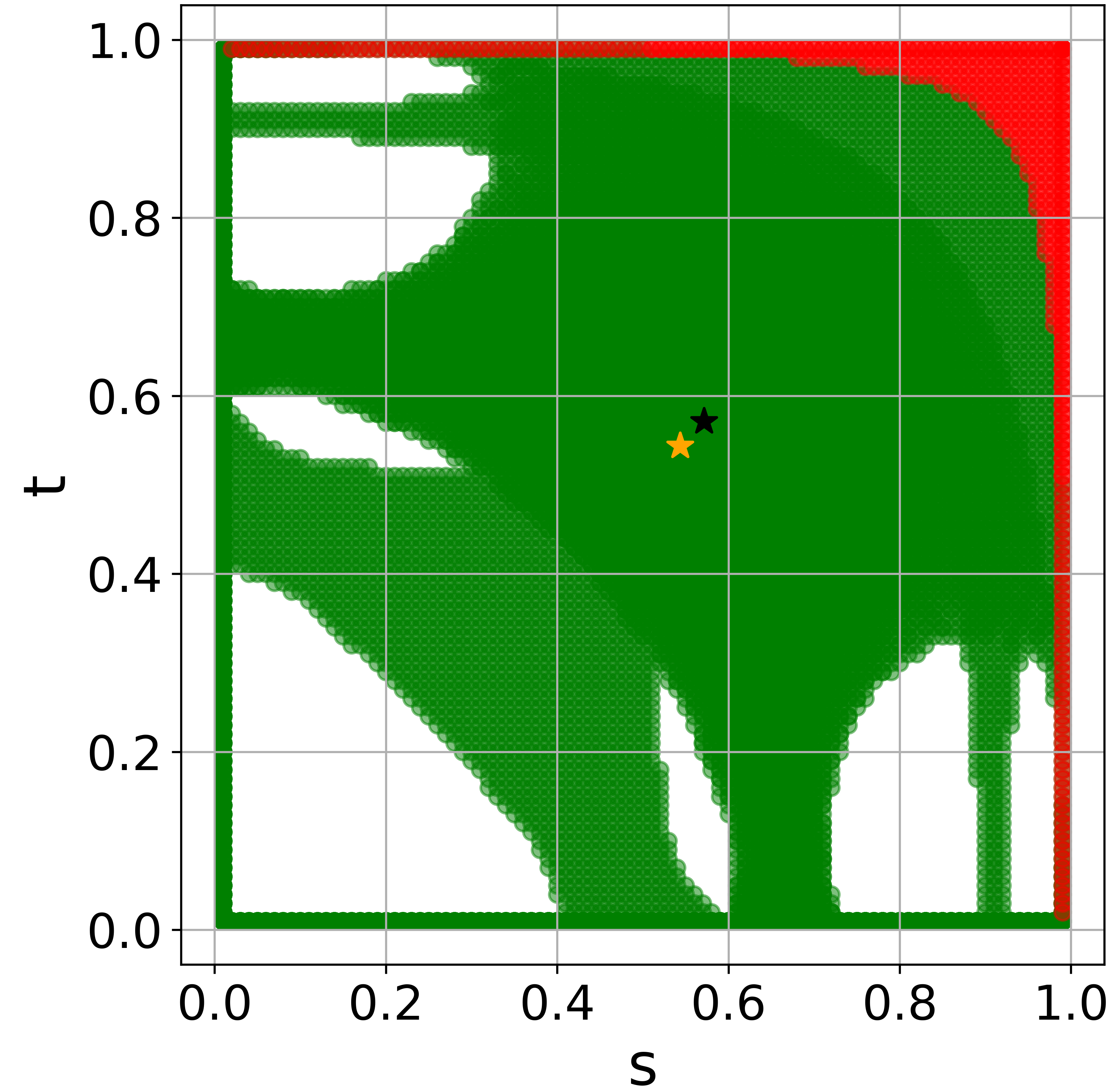}
		\hfill
		\includegraphics[width=0.32\textwidth]{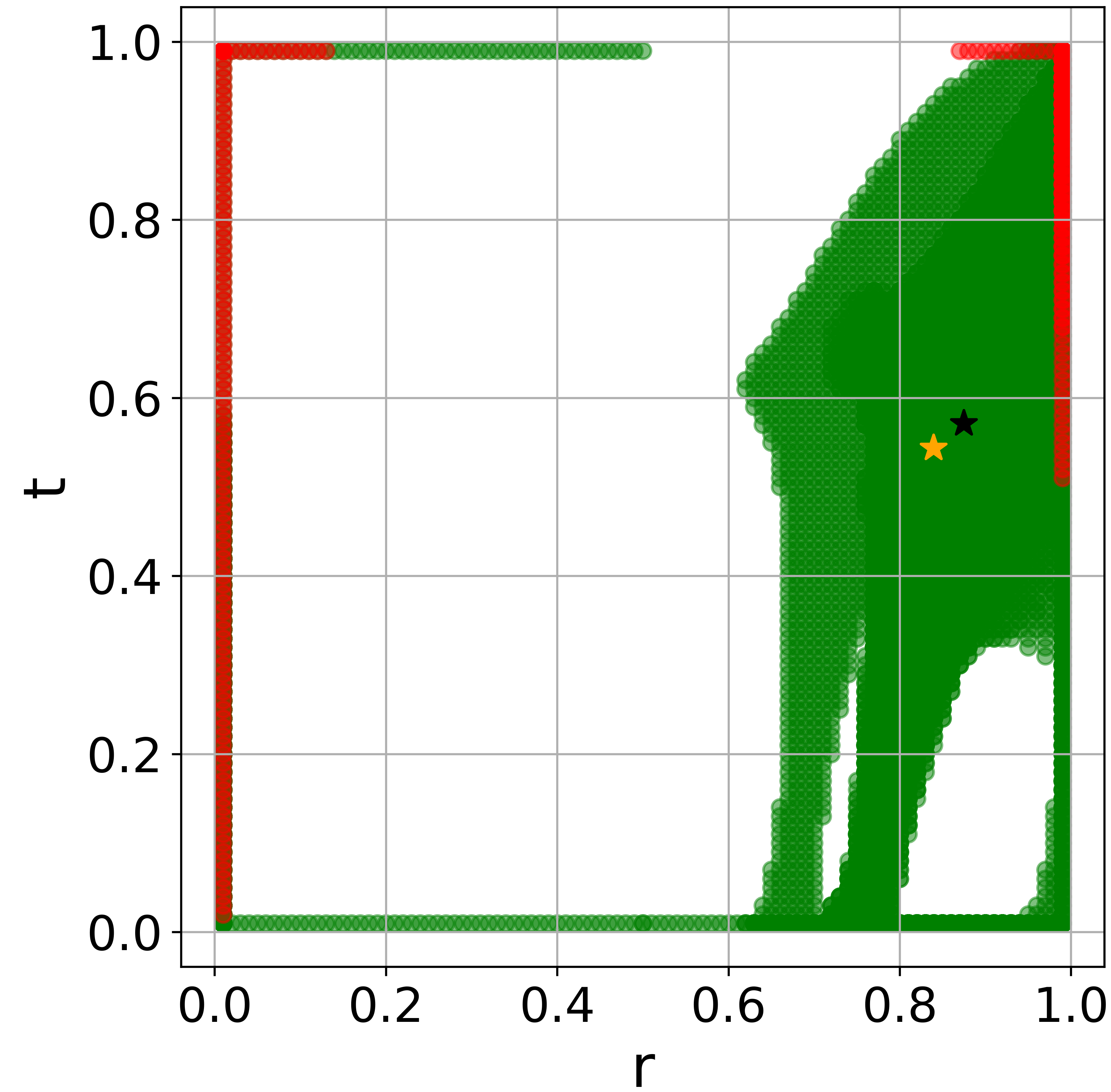}
		\caption{Extremal points of the Convex-Hull of $(\mathbf{r},\mathbf{s},\mathbf{t},\omega^{*})$ projected onto three different $2$-D planes. The computed concave cover $\mathcal{E}^*$ is depicted via the extremal points in these projections. Green (red) points correspond to region where $\mathcal{E}^*>0$ ($\mathcal{E}^*\leq 0$). The points of maximal Hardy violation $(p_H^{opt})$ and maximal Hardy-certified randomness $(\mathcal{R}_{max})$ are shown as yellow and black star markers, respectively. Their extremality implies they are self-tested points of particular interest.}
		\label{fig:chulldia}
	\end{figure*}   
	
	Therefore, based on the considered example and the application of Lemma~\ref{Lemma1}, we conclude that if the black-box probability distribution satisfies $(\mathbf{r},\mathbf{s},\mathbf{t})\in \mathfrak{D}^{\ast}$, then in every subspace of $\mathcal{H}_A^i\otimes \mathcal{H}_B^j\otimes \mathcal{H}_C^k$, it necessarily holds that $r_i=\mathbf{r}$, $s_j=\mathbf{s}$ and $t_k=\mathbf{t}$. This, in turn, enables the self-testing of a class of pure, nonmaximally entangled three-qubit states of the form given by Eq.~(\ref{hardystate}).
	
	Note that, without loss of generality, the phases $\xi$, $\phi$ and $\eta$ can be considered to be independent of the indices $i$, $j$ and $k$ across all subspaces $\mathcal{H}_A^i\otimes \mathcal{H}_B^j\otimes \mathcal{H}_C^k$. This is achieved by appropriately selecting the measurement basis and applying suitable local unitary operations to the state, thereby preserving the observed behaviour $\{p_{m}(a,b,c\vert x,y,z,\rho_m)\}$.
	
	To conclude, suppose that $\vec{\mathfrak{P}}_{Hardy}(\mathbf{r},\mathbf{s},\mathbf{t})$ is observed in a black-box experiment. Lemma~\ref{Lemma1} implies the existence of a tuple $(\mathbf{r},\mathbf{s},\mathbf{t})$ such that, in every subspace, $r_i=\mathbf{r}$, $s_j=\mathbf{s}$ and $t_k=\mathbf{t}$. Consequently, the relation $\vec{P}_{Hardy}(r_i,s_j,t_k)=\vec{\mathfrak{P}}_{Hardy}(\mathbf{r},\mathbf{s},\mathbf{t})$ holds if and only if each subspace state $\rho_{ijk}=\ket{\psi}_{H}\bra{\psi}_{H}$, where $\ket{\psi}_{H}$ is the Hardy state determined by the parameters $(\mathbf{r},\mathbf{s},\mathbf{t})$ as specified in Eq.~(\ref{hardystate}).
	
	Finally, it is always possible to construct local isometries $\Phi^A$, $\Phi^B$ and $\Phi^C$ such that, when applied to the system's Hilbert space and an appropriate ancillary space, they map the unknown state $\ket{\chi}$ of arbitrary dimension onto a known reference Hardy state within a fixed-dimensional ancillary space. Explicitly, this transformation is given by
	\begin{equation}
		\Phi\ket{\chi}^{ABC} \otimes \ket{000}^{A'B'C'}=\zeta^{ABC}\otimes\ket{\psi(r,s,t)}_H^{A'B'C'}   
	\end{equation}
	where $\Phi=\Phi^A\otimes\Phi^B\otimes\Phi^C$, and $\zeta^{ABC}$ represents an auxiliary state associated with the original system.
	
	%%%%%%%%%%%%%%%%%%%%%%%%%%%%%%%%%%%%%%%%%%%%%%%%%%%%%%%%%%%%%%%%%%%%%%%%%%%%%%%%%%%%%%%%%%%%%%%%%%%%%%%%%%%%%%%%%%%%%%%%%%%%%%%%%%%%%%%%%%%%%%%%%%%%%%%%%
	
	\section{Certification of Randomness}\label{randomness}
	
	Our self-testing scheme enables the generation of device-independent certified randomness. We show that whenever tripartite Hardy nonlocality is observed, a nonzero amount of global randomness can be certified.
	
	While randomness certification in the bipartite setting has been widely studied \cite{Pironio2010, Colbeck2012, Colbeck2022, Zhao2023, Sasmal2024}, results in the tripartite case remain scarce. In the simplest bipartite two measurement-two outcomes (2-2-2) scenario, up to $2$-bits of global randomness can be certified using a class of tilted-Bell inequalities \cite{Wooltorton2024}. For the tripartite 3-2-2 case, although the theoretical maximum is $3$-bit, no known protocol achieves this based on Genuine nonlocality. Notably, a maximum Svetlichny inequality violation allows certification of approximately $\qty(3-\log_2(1+\frac{1}{\sqrt{2}}))\approx 2.2284$ bits of randomness. Recently, using Mermin inequality \cite{Mermin1990} authors in \cite{Woodhead2018}, provided tight analytic bound on guessing probabilities associated with one and two of the parties’ measurement outcomes as a function of the Mermin inequality violation. While the maximum violation of Mermin inequality certifies $3$-bit of global randomness, Mermin inequality does not certify genuine nonlocality. 
	
	Here, we demonstrate that the tripartite Hardy argument allows certification of up to $\log_2(7)\approx 2.80734$-bit of randomness—nearing the theoretical limit.
	
	We begin by justifying the certification procedure, showing how Hardy’s nonlocality test in the tripartite setting implies genuine randomness. This can be achieved by re-deriving the tripartite Hardy argument under the assumptions of predictability and nosignalling. While CHSH or Hardy violations rule out predictable nosignalling models in the bipartite case \cite{Masanes2006, Cavalcanti2012, Sasmal2024}, the tripartite case is subtler. By translating the SI/OI framework into the ontological model, we can define predictability and the nosignalling principle across all parties, assuming a reproducible preparation procedure as introduced in \cite{Cavalcanti2012}. A straightforward algebraic evaluation shows that a nonzero Hardy probability implies the failure of predictability in the global joint probability distribution (and hence for local probability also), even if two parties share the strongest form of nosignalling correlation (PR-box), ensuring intrinsic unpredictability. Importantly, this certification is fully device-independent, relying only on observed input-output statistics without requiring any knowledge of the internal workings of the devices.
	
	To estimate the amount of randomness, we must consider a specific quantum system while accounting for a potential adversary (limited by quantum theory) who may control the devices. For unconditional security, we assume the user has no knowledge of the internal workings and can only access the input-output statistics, denoted by $\vec{P}_{obs}$. Since the adversary controls the devices, they may generate $\vec{P}_{obs}$ through convex mixtures of extremal behaviours $\vec{P}_{ext}$ such that $\vec{P}_{obs}=\sum_{ext}q_{ext}\vec{P}_{ext}$, where the adversary knows the weights $q_{ext}$, but the user does not. 
	
	Hence, to securely quantify the amount of randomness $\mathcal{R}(\vec{P}_{obs})$ secure against any quantum side information \cite{Liu2021}, one must evaluate the maximum guessing probability $G(\vec{P}_{obs})=\max_{a_k, x_k} p(\textbf{a}|\textbf{x})$, which is expressed in terms of the min-Entropy to measure the randomness in bits \cite{Konig2008}. Thus, the amount of device-independent randomness pertaining to an observed nonlocal behaviour is given by the following optimisation problem
	\begin{eqnarray} \label{randef}
		&& \mathcal{R}(\vec{P}_{obs}) = -\log_2 \qty[\max_{q_{ext},\vec{P}_{ext}} \sum_{ext} q_{ext}G(\vec{P}_{ext})]  \nonumber \\
		&& \text{Subject to} \nonumber \\
		&& \text{(i)} \ p_{H} = \delta; \ 0<\delta\leq p_{H}^{opt} \nonumber \\
		&& \text{(ii)} p_{AB}(+1,+1 | 1,0,\rho) = p_{BC}(+1,+1 | 1,0,\rho) =0, \nonumber \\
		&& p_{AC}(+1,+1 | 0,1,\rho) =p(-1,-1,-1 |1,1,1,\rho) = 0. \nonumber \\
		&& \text{(iii)} \ \vec{P}_{obs} = \sum_{ext}q_{ext}\vec{P}_{ext} \in \mathcal{Q}
	\end{eqnarray} 
	where $G(\vec{P}_{ext}) = \max_{a_k, x_k} p_{ext}(\textbf{a}|\textbf{x})$, $p_{H}^{opt}=0.0181$ is the quantum optimal value of the Hardy probability and $\mathcal{Q}$ is the set of quantum behaviours. This optimisation is generally difficult, as quantum theory does not form a polytope. However, it can be approximated using SDP based approach known as the NPA hierarchy \cite{Navascues2007}, which provides upper bounds on the adversary’s guessing probability. The implementation of NPA-hierarchy \cite{gitmishra} was done
	in PYTHON using CVXPY (see Appx~\ref{apnB}) and the SDP computed randomness is presented in Fig.~\ref{fig:rand} as a function of Hardy probability. 
	
	Alternatively, if the observed behaviour is extremal and self-testable, the optimisation becomes trivial. By identifying a region in the parameter space $(r,s,t)$ where Hardy's argument enables self-testing, even without maximal violation, these behaviours are extremal and uniquely determined. Thus, the guessing probability can be directly computed. Unlike the SDP-computed randomness, randomness evaluated from the self-tested quantum strategy generates a region of certified randomness. Different combinations of parameters $r$, $s$, and $t$ can yield the same Hardy nonmaximal violation ($p_H$), with each nonzero $p_H$ producing a corresponding region of certified randomness as illustrated in Fig.~\ref{fig:rand}.
	
	To analyse the maximum certifiable randomness via the tripartite Hardy relations, we closely examine the structure of the Hardy behaviour defined in Eq.~(\ref{hardybehavtable}). Ideally, maximal randomness corresponds to a uniform distribution over eight outcomes $(\frac{1}{8})$. However, such a case is unattainable due to the presence of structural zeroes in all but one row of the distribution. 
	
	Consider the first row, corresponding to the measurement setting $(x,y,z)=(0,0,0)$, i.e. the distribution $\{p(a,b,c|0,0,0,\rho)\}$. Here, the first entry attains a maximum probability of $p_{H}=0.0181 < \frac{1}{8}$. A straightforward calculation shows that the remaining probabilities cannot be uniformly distributed as $\frac{1}{7}(1-0.0181)$, ruling out the possibility of achieving a uniform distribution in this row.
	
	The next best candidate for uniformity arises from the last row, corresponding to the setting $(x,y,z)=(1,1,1)$, where the final element $p(-1,-1,-1|1,1,1,\rho)=0$. In this scenario, maximal randomness is attained when the remaining seven outcomes each occur with equal probability $\frac{1}{7}$. This distribution becomes feasible when the parameters are set to $s=t=\frac{4}{7}$ and $r=\frac{1}{2s}\qty(3-\sqrt{\frac{4}{s}-3})=\frac{7}{8}$. 
	
	For the parameter set $(\frac{7}{8},\frac{4}{7},\frac{4}{7})$, the Hardy probability is $p_{H}(\frac{7}{8},\frac{4}{7},\frac{4}{7})=0.0179$, and the corresponding amount of certified randomness is $\mathcal{R}(A_1,B_1,C_1)=\log_2 7 \approx 2.8073$ bits. Crucially, this point lies on an extremal boundary of the convex hull as shown in Fig.~\ref{fig:chulldia}, and it is also a self-testable point—thereby confirming that the quantified randomness is genuinely certified.
	
	\begin{figure}[htbp]
		\centering
		\includegraphics[width=0.48\textwidth]{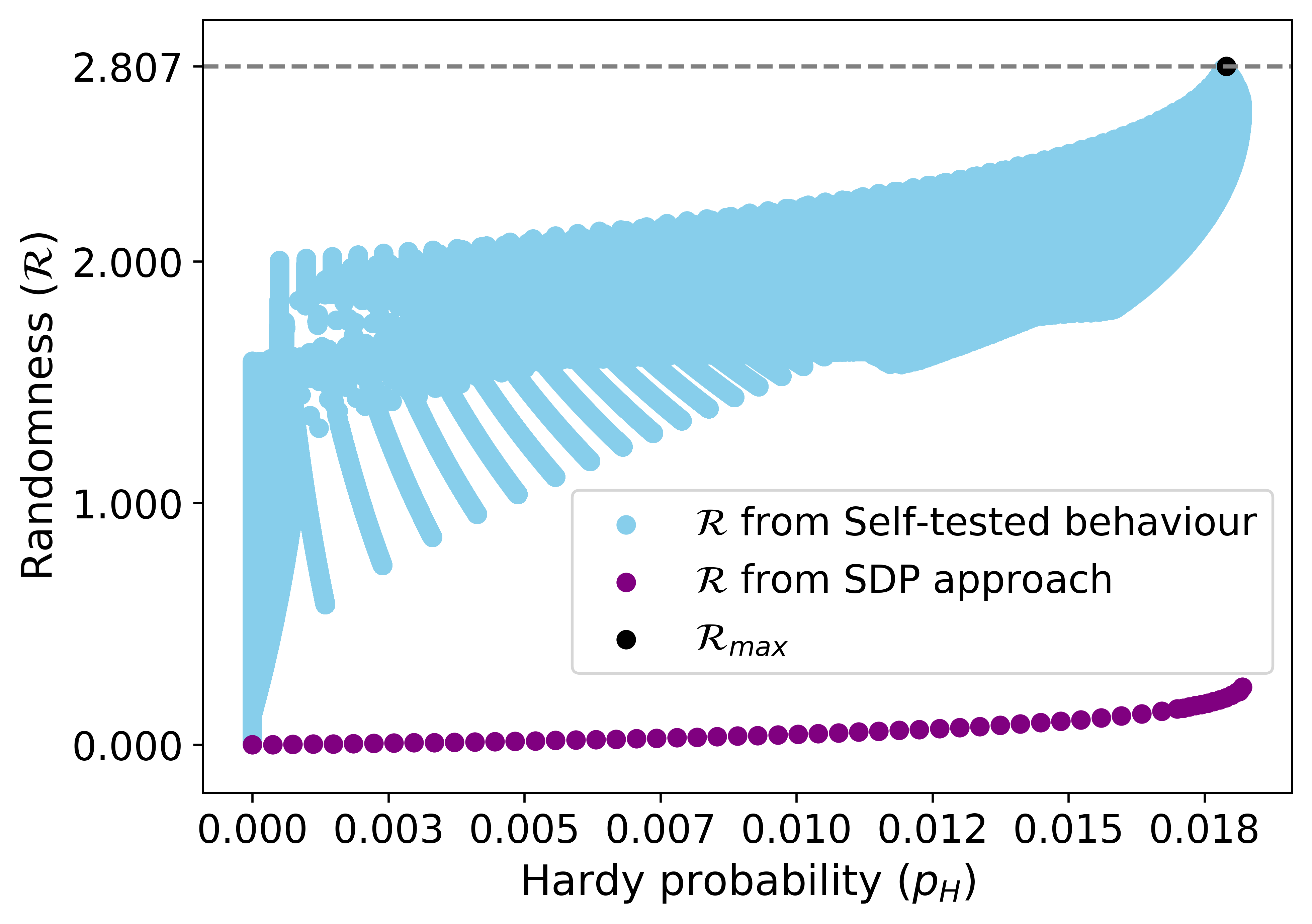}
		\caption{Plot illustrating two types of Hardy-certified randomness: (i) a region of randomness (sky-blue shaded areas) obtained from the self-tested quantum strategies, and (ii) a monotonic curve (purple dotted line) corresponding to the randomness computed via the NPA hierarchy through SDP in Python. Owing to the fact that different combinations of parameters ($r$, $s$, $t$) can lead to the same Hardy violation ($p_H$), each nonzero value of $p_H$ generates a region of certified randomness from the self-tested behaviours. In contrast, the SDP approach yields a unique value of certified randomness for each $p_H$. The SDP computation achieves a maximum randomness of approximately $0.2387$ bits at the optimal Hardy probability $p_H^{\mathrm{opt}} = 0.0181$, whereas the maximum randomness certified from the self-tested behaviour reaches $\log_2 7 \approx 2.8073$ bits when $p_H = 0.0179$.}
		\label{fig:rand}
	\end{figure}

	%%%%%%%%%%%%%%%%%%%%%%%%%%%%%%%%%%%%%%%%%%%%%%%%%%%%%%%%%%%%%%%%%%%%%%%%%%%%%%%%%%%%%%%%%%%%%%%%%%%%%%%%%%%%%%%%%%%%%%%%%%%%%%%%%%%%%%%%%%%%%%%%%%%%%%%%%%%%%%%%%%%%%%%%%%%%%%%%%%%%%%%%%%%%%%%%%%%%%%%%%%%%%%%%%%%%%%%%%%%%%%%%%%%%%%%%%%%%%%%%%%%%%%%%%%%%%%%%%%%%%%%%%%%%%%%%%%%%%%%%%%%%%%%%%%%%%%%%%%%%%%%%%%

	\section{Summary and Outlook}\label{conclu}
	
	We develop a framework for self-testing a broad class of pure nonmaximally genuinely entangled tripartite states using Hardy relations. We begin by showing that tripartite Hardy-type arguments \cite{Rahaman2014} effectively capture genuine nonlocality. Specifically, we demonstrate that any correlation exhibiting a nonzero Hardy probability, while satisfying the corresponding constraints, cannot be explained by a fully local model, nor by any ontic model in which two parties share even unphysical PR box \cite{Popescu1994} type of correlations.
	
	Analysing these ontic models through the assumptions of SI and OI \cite{Jarrett1984, Shimony1993}, we find that a fully local ontic model \cite{Mermin1990} require the mutual satisfaction of both SI and OI among all parties. In contrast to fully local ontic model, a nosignalling bilocal model \cite{Gallego2012, Bancal2013} satisfies SI globally (similar to satisfying nosignalling principle in operational level) while relaxing OI between any two parties, effectively allowing stronger nosignalling correlations such as PR-box-type bipartite correlations. Although Hardy relations were originally viewed within the framework of fully local ontic models, our analysis reveals that they detect a stronger form of nonlocality than that evidenced by conventional Mermin-type violations \cite{Mermin1990}.
	
	This distinction is crucial because the violation of a fully local model, such as through Mermin inequalities, does not preclude alternative ontic models with stronger bipartite correlations. Thus, Mermin-type violations may arise from either genuinely entangled states or specific ontic models that relax OI between parties.
	
	A key advantage of Hardy-based self-testing over traditional Bell inequality-based approaches lies in its ability to certify quantum behaviour even under nonmaximal violations. Bell-based frameworks typically allow self-testing only at specific points of maximal inequality violation. In contrast, Hardy-type arguments define a family of joint probability distributions satisfying four zero constraints. These conditions significantly reduce the dimensionality of the $26$-dimensional probability space into $13$-dimensional space and enable full characterisation of quantum correlations using just three parameters, which manifest in both the state and measurements.
	
	We implement a concave cover construction using a convex-hull approach in Python to identify regions of these parameters that uniquely realise specific quantum correlations. This permits self-testing of the state (up to local isometry and complex conjugation) even for nonmaximally Hardy-violating behaviours, thereby broadening the set of certifiable quantum strategies and enhancing the flexibility of Hardy-based methods.
	
	In addition, we demonstrate that Hardy-type correlations can be used to certify genuine randomness. Employing a SDP technique implemented in Python (cvxpy), we compute lower bounds on the extractable randomness and then compare the certified randomness across various self-tested points. Unlike Bell functionals, which typically self-test a single extremal point per inequality, Hardy-type arguments identify an entire set of behaviours corresponding to a fixed nonzero Hardy probability—all of which are self-testable. This enables the certification of a region of randomness for a given Hardy violation. We find that the maximal randomness certified using Hardy correlations can reach up to $\log_2 7\approx 2.8073$ bits.
	
	Our findings also open up several promising directions for future research. Generalising the SI/OI-based framework to four or more parties could uncover deeper structural features of multipartite nonlocality. In our analysis of Lemma \ref{Lemma1}, we considered the Hardy probability function to define a sub-region within the state parameter space; exploring alternative linear combinations of probabilities may extend this region further and improve the robustness of the self-testing procedure. Additionally, investigating Hardy-type constructions with fewer than four zero constraints could potentially allow for larger randomness certification, possibly approaching the theoretical tripartite bound of $3$-bits. Our methodology also paves the way for analogous studies in systems with more than three parties.
	
	Finally, it is worth noting that in the bipartite case, Hardy correlations lie on the nosignalling facet and are extremal but non-exposed within the quantum set \cite{Chen2023}. Extending this geometric analysis to the tripartite setting could offer richer insight into the boundary structure of the quantum set beyond the well-studied simplest $2-2-2$ scenario.
	
	In conclusion, while various logical nonlocality arguments exist for multipartite systems, the question of self-testing in each case remains both fundamental and practically significant, particularly for applications in quantum information processing where device-independent verifiability plays a central role.
	
	%%%%%%%%%%%%%%%%%%%%%%%%%%%%%%%%%%%%%%%%%%%%%%%%%%%%%%%%%%%%%%%%%%%%%%%%%%%%%%%%%%%%%%%%%%%%%%%%%%%%%%%%%%%%%%%%%%%%%%%%%%%%%%%%%%%%%%%%%%%%%%%%%%%%%%%%%%%%%%%%%%%%%%%%%%%%%%%%%%%%%%%%%%%%%%%%%%%%%%%%%%%%%%% %%%%%%%%%%%%%%%%%%%%%%% 
	
	\section*{Acknowledgments}
	We are grateful for stimulating discussions with Ashutosh Rai, Subhendu Bikash Ghosh, and Snehasish Roy Chowdhury. R.A. acknowledges financial support from the Council of Scientific and Industrial Research (CSIR), Government of India. S.S. acknowledges support from the National Natural Science Foundation of China (Grant No. G0512250610191).

	%%%%%%%%%%%%%%%%%%%%%%%%%%%%%%%%%%%%%%%%%%%%%%%%%%%%%%%%%%%%%%%%%%%%%%%%%%%%%%%%%%%%%%%%%%%%%%%%%%%%%%%%%%%%%%%%%%%%%%%%%%%%%%%%%%%%%%%%%%%%%%%%%%%%%%%%%%%%%%%%%%%%%%%%%%%%%%%%%%%%%%%%%%%%%%%%%%%%%%%%%%%%%%%%%%%%%%%%%%%%%%%%%%%%%%%%%%%%%%%%%%%%%%%%%%%%%%%%%%%%%%%%%%%%%%
	
	%\bibliography{reference}
	

	%%%%%%%%%%%%%%%%%%%%%%%%%%%%%%%%%%%%%%%%%%%%%%%%%%%%%%%%%%%%%%%%%%%%%%%%%%%%%%%%%%%%%%%%%%%%%%%%%%%%%%%%%%%%%%%%%%%%%%%%%%%%%%%%%%%%%%%%%%%%%%%%%%%%%%%%%%%%%%%%%%%%%%%%%%%%%%%%%%%%%%%%%%%%%%%%%%%%%%%%%%%%%%%%%%%%%%%%%%%%%%%%%%%%%%%%%%%%%%%%%%%%%%%%%%%%%%%%%%%%%%%%%%%%%%
	
	\appendix

	\onecolumngrid

	\section{Derivation of Hardy relation from the assumptions of SI and OI in ontic level} \label{apnA}
	
	Consider the joint probability distribution $p(a,b,c|x,y,z,\lambda)$ in the ontic model satisfying $SI(3,3)\wedge OI(3,3)$ given by Eqs.~(\ref{si33}) and (\ref{oi33}).
	\begin{equation}
		\begin{aligned}
			p(a,b,c|x,y,z,\lambda) &= p(a|x,y,z,\lambda)p(b,c|a,x,y,z,\lambda) \ \ \qty[\text{Definition of joint probability} ]\\
			&= p(a|x,y,z,\lambda)p(b|a,x,y,z,\lambda)p(c|a,b,x,y,z,\lambda) \ \ \qty[\text{Definition of joint probability} ]\\
			&=p(a|x,\lambda)p(b|a,y,\lambda)p(c|a,b,z,\lambda) \ \ \qty[SI(3,3)] \\
			&=p(a|x,\lambda)p(b|y,\lambda)p(c|z,\lambda) \ \ \qty[OI(3,3): OI(A\rightleftharpoons B,A\rightleftharpoons C,B\rightleftharpoons C)]
		\end{aligned}
	\end{equation}
	Therefore, under the assumptions $SI(3,3)\wedge OI(3,3)\implies$ the ontic model leads to \textit{full-factorisability}. 
	
	We now demonstrate that, under the constraints from Eq.~(\ref{tri}), the above assumptions imply $p_H=p(+1,+1,+1|0,0,0,\lambda)=p_A(+1|0,\lambda)p_B(+1|0,\lambda)p_C(+1|0,\lambda)=0$. The set of Hardy-constraints (employing the factorisability condition), are
	\begin{eqnarray}
		p_{AB}(+1,+1 | 1,0,\lambda) &=& p_{A}(+1| 1,\lambda)p_{B}(+1 |0,\lambda)= 0 \label{c1}\\
		p_{BC}(+1,+1 | 1,0,\lambda) &=& p_{B}(+1| 1,\lambda)p_{C}(+1 |0,\lambda)= 0 \label{c2}\\
		p_{AC}(+1,+1 | 0,1,\lambda) &=& p_{A}(+1| 0,\lambda)p_{C}(+1 |1,\lambda)= 0 \label{c3}\\
		p(-1,-1,-1 |1,1,1,\lambda) &=& p_{A}(-1| 1,\lambda)p_{B}(-1 |1,\lambda)p_{C}(-1|1,\lambda)= 0 \label{c4}
	\end{eqnarray} 
	From Eqs.~(\ref{c1}-\ref{c3}), we derive
	\begin{eqnarray}
		\text{(\ref{c1})} &\implies& \qty[1-p_{A}(-1| 1,\lambda)]p_{B}(+1 |0,\lambda)=0 \implies p_{B}(+1 |0,\lambda)=p_{B}(+1 |0,\lambda)p_{A}(-1| 1,\lambda) \label{c5} \\
		\text{(\ref{c2})} &\implies& \qty[1-p_{B}(-1| 1,\lambda)]p_{C}(+1 |0,\lambda)=0 \implies p_{C}(+1 |0,\lambda)=p_{C}(+1 |0,\lambda)p_{B}(-1 |1,\lambda) \label{c6}\\
		\text{(\ref{c3})} &\implies& \qty[1-p_{c}(-1| 1,\lambda)]p_{A}(+1 |0,\lambda)=0 \implies p_{A}(+1 |0,\lambda)=p_{A}(+1 |0,\lambda)p_{C}(-1| 1,\lambda) \label{c7}
	\end{eqnarray}
	Multiplying Eqs.~(\ref{c5}), (\ref{c6}) and (\ref{c7}) yields
	\begin{eqnarray}
		&&p_B(+1|0,\lambda)p_C(+1|0,\lambda)p_A(+1|0,\lambda)=p_{B}(+1 |0,\lambda)p_{A}(-1| 1,\lambda)p_{C}(+1 |0,\lambda)p_{B}(-1 |1,\lambda)p_{A}(+1 |0,\lambda)p_{C}(-1| 1,\lambda) \nonumber \\
		&\implies& p_H=\Big[ p_{A}(-1| 1,\lambda)p_{B}(-1| 1,\lambda)p_{C}(-1| 1,\lambda) \Big] p_B(+1|0,\lambda)p_C(+1|0,\lambda)p_A(+1|0,\lambda) =0 \ \ \qty[\text{From Eq.~(\ref{c4})}]
	\end{eqnarray}
	Therefore, any fully local ontic model satisfying $SI(3,3)\wedge OI (3,3)$ must yield $p_H=0$. While this result has been previously established \cite{Rahaman2014}, we highlight here that Hardy’s correlations reveal a stronger form of nonlocality than Mermin-type correlations-an aspect that has not been fully appreciated.
	
	To illustrate this further, consider a scenario in which outcome independence between Bob and Charlie is relaxed. That is, we adopt the assumption $SI(3,3)\wedge OI(2,2)$. In quantum theory, such a scenario could correspond to Bob and Charlie sharing a maximally entangled state, implying that their outcomes are pre-correlated. This is particularly relevant because there exist correlations that cannot be explained by fully local models but can be simulated by ontic models where some parties share entangled states. If we do not confine ourselves to quantum theory but allow any nosignalling correlations, Bob and Charlie could share PR-box correlations \cite{Popescu1994}. 
	
	In standard literature, such ontic models are often termed as NSBL models \cite{Bancal2013}. In our framework, this corresponds to an ontic model with assumption $SI(3,3)\wedge OI(2,2)$. In such a ontic model, considering the joint probability distribution $p(a,b,c|x,y,z,\lambda)$
	\begin{equation}
		\begin{aligned}
			p(a,b,c|x,y,z,\lambda) &= p(a|x,y,z,\lambda)p(b,c|a,x,y,z,\lambda) \ \ \qty[\text{Definition of joint probability} ]\\
			&= p(a|x,y,z,\lambda)p(b|a,x,y,z,\lambda)p(c|a,b,x,y,z,\lambda) \ \ \qty[\text{Definition of joint probability} ]\\
			&=p(a|x,\lambda)p(b|a,y,\lambda)p(c|a,b,z,\lambda) \ \ \qty[SI(3,3)] \\
			&=p(a|x,\lambda)p(b|y,\lambda)p(c|b,z,\lambda) \ \ \qty[OI(2,2): OI(A\rightleftharpoons B,A\rightleftharpoons C)] \\
			&=p(a|x,\lambda)p(b,c|y,z,\lambda) \ \ \qty[\text{using} \ p(c|b,z,\lambda)=\frac{p(b,c|y,z,\lambda)}{p(b|y,\lambda)}]
		\end{aligned}
	\end{equation}
	The set of Hardy-constraints (applying this new bi-factorisability condition) are now becomes
	\begin{eqnarray}
		p_{AB}(+1,+1 | 1,0,\lambda) &=& p_{A}(+1| 1,\lambda)p_{B}(+1 |0,\lambda)= 0 \label{cc1} \\
		p_{BC}(+1,+1 | 1,0,\lambda) &=& p_{B}(+1| 1,\lambda)p_{C}(+1 |+1_b,0,\lambda)= 0 \label{cc2} \\
		p_{AC}(+1,+1 | 0,1,\lambda) &=& p_{A}(+1| 0,\lambda)p_{C}(+1 |1,\lambda)= 0 \label{cc3} \\
		p(-1,-1,-1 |1,1,1,\lambda) &=& p_{A}(-1| 1,\lambda)p_{B}(-1 |1,\lambda)p_{C}(-1|-1_b,1,\lambda)= 0 \label{cc4}
	\end{eqnarray} 
	From Eqs.~(\ref{cc1}-\ref{cc3}), similar derivations lead to
	\begin{eqnarray}
		\text{(\ref{c1})} &\implies& \qty[1-p_{A}(-1| 1,\lambda)]p_{B}(+1 |0,\lambda)=0 \implies p_{B}(+1 |0,\lambda)=p_{B}(+1 |0,\lambda)p_{A}(-1| 1,\lambda) \label{cc5} \\
		\text{(\ref{c2})} &\implies& \qty[1-p_{B}(-1| 1,\lambda)]p_{C}(+1 |+1_b,0,\lambda)=0 \implies p_{C}(+1 |+1_b,0,\lambda)=p_{B}(-1 |1,\lambda)p_{C}(+1 |+1_b,0,\lambda) \label{cc6}\\
		\text{(\ref{c3})} &\implies& \qty[1-p_{c}(-1| 1,\lambda)]p_{A}(+1 |0,\lambda)=0 \implies p_{A}(+1 |0,\lambda)=p_{A}(+1 |0,\lambda)p_{C}(-1| 1,\lambda) \label{cc7}
	\end{eqnarray}
	Multiplying Eqs.~(\ref{cc5}), (\ref{cc6}) and (\ref{cc7}) leads to
	\begin{eqnarray}
		&&p_B(+1|0,\lambda)p_C(+1|+1_b,0,\lambda)p_A(+1|0,\lambda)=p_{B}(+1 |0,\lambda)p_{A}(-1| 1,\lambda)p_{B}(-1 |1,\lambda)p_{C}(+1 |+1_b,0,\lambda)p_{A}(+1| 0,\lambda)p_{C}(-1| 1,\lambda) \nonumber \\
		&&\implies p_H=\Big[ p_{A}(-1| 1,\lambda)p_{B}(-1| 1,\lambda)p_{C}(-1| 1,\lambda) \Big] p_B(+1|0,\lambda)p_C(+1|+1_b,0,\lambda)p_A(+1|0,\lambda) \nonumber \\
		&&\implies p_H\Big[1-p_{A}(-1| 1,\lambda)p_{B}(-1| 1,\lambda)p_{C}(-1| 1,\lambda)\Big]=0 \nonumber \\
		&&\implies p_H=0, \ \ \text{Or} \ \ p_{A}(-1| 1,\lambda)p_{B}(-1| 1,\lambda)p_{C}(-1| 1,\lambda)=1 \nonumber \\
		&&\implies p_H=0, \ \ \text{Or} \ \ p_{A}(-1| 1,\lambda)=p_{B}(-1| 1,\lambda)=p_{C}(-1| 1,\lambda)=1
	\end{eqnarray}
	It is crucial to observe the implications of the condition $p_{C}(-1| 1,\lambda)=1$. This means that Charlie's outcome for the input setting $z=1$ is fully deterministic, regardless of any other variables. In other words, once we know $z=1$, the value of $c=-1$ becomes fixed, even if we condition on additional hidden variables or on Bob's output. Consequently, we can write $p_{C}(-1| b,1,\lambda)=1$ without loss of generality. 
	
	Thus, along with $p_{A}(-1| 1,\lambda)=p_{B}(-1| 1,\lambda)=1$ imply that $p_{A}(-1| 1,\lambda)p_{B}(-1| 1,\lambda)p_{C}(-1| b,1,\lambda)=1$ directly contradicts the constraint given by Eq.~(\ref{cc4}). Therefore, the only consistent solution is $p_H=0$. This implies that the Hardy probability must vanish even when we \textit{relax} the assumption of outcome independence between Bob and Charlie, i.e., no ontic model satisfying $SI(3,3)\wedge OI(2,2)$ can reproduce a nonzero Hardy probability.
	
	Hence, Hardy’s test of nonlocality detects a form of nonlocal correlation that is strictly stronger than that revealed by Mermin-type inequalities, which only require the violation of full locality. The Hardy argument goes beyond this by ruling out even a broader class of ontic models, including those that allow some outcome dependence (NSBL models). In this sense, under the assumption of full settings independence $SI(3,3)$, Hardy-type violating correlations qualify as genuinely nonlocal.

	%%%%%%%%%%%%%%%%%%%%%%%%%%%%%%%%%%%%%%%%%%%%%%%%%%%%%%%%%%%%%%%%%%%%%%%%%%%%%%%%%%%%%%%%%%%%%%%%%%%%%%%%%%%%%%%%%%%%%%%%%%%%%%%%%%%%%%%%%%%%%%%%%%%%%%%%%%%%%%%%%%%%%%%%%%%%
	
	\section{Computation of Hardy-certified randomness in quantum theory using NPA hierarchy} \label{apnB}
	
	The amount of Hardy-certified randomness in quantum theory for a given nonzero value of $p_{Hardy}=\delta$, denoted by $\mathcal{R}(\delta)$, is determined by solving the following optimisation problem
	\begin{equation} \label{qopt}
		\begin{aligned}
			&\mathcal{R}(\delta) =-\log_{2}\qty[\qty(\max\limits_{\{\rho,A_{a|x},B_{b|y},C_{c|z}\}}~\Tr[\rho A_{a|x} \otimes B_{b|y} \otimes C_{c|z}])]  \\
			& \text{Subject to} \\
			& \text{(i)} \ p_{H} = \delta; \ 0<\delta\leq p_{H}^{opt}  \\
			& \text{(ii)} p_{AB}(+1,+1 | 1,0,\rho) = p_{BC}(+1,+1 | 1,0,\rho) = p_{AC}(+1,+1 | 0,1,\rho) =p(-1,-1,-1 |1,1,1,\rho) = 0.
		\end{aligned}
	\end{equation}
	It is important to note that the above optimisation is defined over both the quantum state and the measurement operators, making it a non-convex problem—one that is, in general, computationally challenging. Nevertheless, this non-convex problem can be transformed into a convex optimisation problem (SDP) through an appropriate relaxation. To achieve this, we consider optimising over a convex set $S$ which strictly contains the quantum set $\mathcal{Q}$, thereby enabling a computable formulation. In particular, we employ a relaxation method known as the NPA hierarchy \cite{Navascues2007,Navascues2008}, which provides an infinite sequence of outer approximations of the quantum set $
	\mathcal{Q}_1 \supset \mathcal{Q}_2 \supset \dots \supset \mathcal{Q}_n \supset \dots$. Each level in this hierarchy is defined via SDP. It has been proven that these sets converge to the quantum set in the limit $ n \to \infty $, i.e., $\lim_{n\to\infty}\mathcal{Q}_n = \mathcal{Q}$ \cite{Navascues2007,Navascues2008}. Now, for our purposes, we choose $ S = \mathcal{Q}_3 $. Consequently, the optimisation problem in Eq.~(\ref{qopt}) becomes
	\begin{equation}
		\begin{aligned}
			&\mathcal{R}(\delta) =-\log_{2}\qty[\max\limits_{\vec{P_{\delta}}\in \mathcal{Q}_3}\qty(\max\limits_{a,b,c,x,y,z}~p(a,b,c|x,y,z))] \\
			& \text{Subject to} \\
			& \text{(i)} \ p_{H} = \delta; \ 0<\delta\leq p_{H}^{opt}  \\
			& \text{(ii)} p_{AB}(+1,+1 | 1,0,\rho) = p_{BC}(+1,+1 | 1,0,\rho) = p_{AC}(+1,+1 | 0,1,\rho) =p(-1,-1,-1 |1,1,1,\rho) = 0.
		\end{aligned}
	\end{equation}
	The above optimisation problem has been implemented in PYTHON using CVXPY, by setting the hierarchy level as $n=3$ \cite{gitmishra}. It is important to note that, by construction, the certified randomness over the full quantum set $\mathcal{Q}$ satisfies $
	\mathcal{R}^{\mathcal{Q}}(\text{Hardy}) \ge \mathcal{R}^{\mathcal{Q}_3}(\text{Hardy})$.

\end{document}